\documentclass{article}
\usepackage{PRIMEarxiv}

\usepackage{mymacros}
\usepackage{natbib}

\usepackage[utf8]{inputenc} 
\usepackage[T1]{fontenc}    
\usepackage{hyperref}       
\usepackage{url}            
\usepackage{booktabs}       
\usepackage{amsfonts}       
\usepackage{nicefrac}       
\usepackage{microtype}      
\usepackage{xcolor}         

\title{High confidence inference on the probability an individual benefits from treatment using experimental or observational data}

\author{%
  Gabriel Ruiz\thanks{\url{gabriel-ruiz-github.io}} \\
  Department of Statistics\\
  University of California, Los Angeles\\
  Los Angeles, CA 90095\\
  \texttt{\url{17ruiz17@gmail.com}} \\
  \AND
  Oscar Hernan Madrid Padilla\thanks{\url{https://hernanmp.github.io/}}\\
  Department of Statistics\\
  University of California, Los Angeles\\
  Los Angeles, CA 90095\\
  \texttt{\url{oscar.madrid@stat.ucla.edu}}\\
}

\begin{document}

\maketitle

\begin{abstract}
We seek to understand the probability an individual benefits from treatment (PIBT), an inestimable quantity that must be bounded in practice. Given the innate uncertainty in the population-level bounds on PIBT, we seek to better understand the margin of error for their estimation in order to discern whether the estimated bounds on PIBT are tight or wide due to random chance or not. Toward this goal, we present guarantees to the estimation of bounds on marginal PIBT, with any threshold of interest, for a randomized experiment setting or an observational setting where propensity scores are known. We also derive results that permit us to understand heterogeneity in PIBT across learnable sub-groups delineated by pre-treatment features. These results can be used to help with formal statistical power analyses and frequentist confidence statements for settings where we are interested in assumption-lean inference on PIBT through the target bounds with minimal computational complexity compared to a bootstrap approach. Through a real data example from a large randomized experiment, we also demonstrate how our results for PIBT can allow us to understand the practical implication and goodness of fit of an estimate for the conditional average treatment effect (CATE), a function of an individual's baseline covariates. 
\end{abstract}

\section{Introduction}
\subsection{Motivation}

The premise of this work, in a vein similar to predictive inference with quantile regression, is that observations may lie far away from their conditional expectation. In the context of causal inference, due to the missing-ness of one outcome, it is difficult to check whether an individual's treatment effect lies close to its prediction given by the estimated average treatment effect or conditional average treatment effect ($\cate$). With the aim of augmenting the inference with these estimands in practice, we further study a distribution-free framework for the plug-in estimation of bounds on the probability an individual benefits from treatment ($\pibt$), a generally inestimable quantity that would concisely summarize an intervention's efficacy if it could be known.

In order to study $\pibt$, we are interested in the individual treatment effect:
\[
Y_i(1)-Y_i(0)\ (i=1,\dots,n).
\]
Here, $Y_i(w)$ ($w=0,1$) is known as a potential outcome  \citep{splawa-neyman1990,Rubin1974EstimatingCE} and equivalently as a counterfactual \citep{pearlCausality,hernan2020}. The random variable $Y_i(w)$ corresponds to the outcome $Y_i$ of interest when a hypothetical experimenter intervenes with nature to force the binary exposure indicator $W_i$ for individual $i$ to be little $w\in\{0,1\}$, typically through random assignment in an experiment. This intervention can also be denoted by \citet{pearlCausality}'s ``do'' operator: $Y_i(w)$ is the outcome we observe under $\text{do}(W_i=w)$. 

Suppose a large value of the observed outcome $Y_i=W_iY_i(1)+(1-W_i)Y_i(0)$ is ``good'' for an individual. Assuming we wish to show $W_i=1$ is effective at accomplishing this, we will consider an individual such that 
\begin{equation}\label{eqn:benefitTrtmntThresh}
Y_i(1)-Y_i(0)>\delta
\end{equation}
to have benefited from treatment. We can take $\delta$ to be any relevant threshold we would like, such as $\delta=0$. We may also reverse the inequality if more appropriate. Just as well, the vocabular and notational semantics of the inequality in \eqref{eqn:benefitTrtmntThresh} can instead be with respect to whether an individual is harmed by an intervention as we discuss in Appendix \ref{sec:PIHI}. Should the potential outcomes be binary, we may also use the inequality in \eqref{eqn:benefitTrtmntThresh} with $\delta=0$ as we discuss in Appendix \ref{sec:binaryOutcomes}. Moreover, if the outcome of interest is strictly positive, we may alternatively define benefiting from treatment in terms of the ratio of an individual's potential outcomes being above a threshold as we discuss in Appendix \ref{sec:diffDefnRatio}. For simplicity of presentation and without any loss of generality for these alternative definitions which our results can be applied to, we will say an individual benefits from treatment when the inequality in \eqref{eqn:benefitTrtmntThresh} holds. 
\begin{figure}[htpb]
\centering
\includegraphics[width=0.45\textwidth]{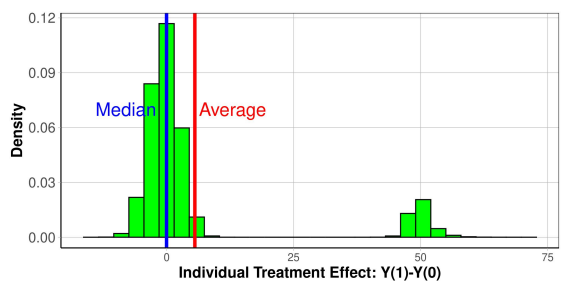}
\caption{A hypothetical distribution for the individual treatment effects. Here, the mean is positive yet the probability an individual's treatment outcome is better (larger in value) than their control outcome is 49.75\%.\label{fig:hypothITE}}
\end{figure}

As can be appreciated from the hypothetical distribution of the individual treatment effects in Figure \ref{fig:hypothITE}, it is very well possible that the average of the individual treatment effect distribution is pulled by outliers and therefore may mislead us. This example motivates our interest in the probability an individual benefits from treatment ($\pibt$) as provided in Definition \ref{defn:propBenefit}.

\begin{definition}[Probability an individual benefits from treatment ($\pibt$)\label{defn:propBenefit}]

To understand heterogeneity in treatment effect for individuals across differing strata of a pre-treatment covariate, $X_i$, denote
\[
{\rm pr}\{Y_i(1)-Y_i(0)>\delta\mid X_i=x\}
\]
as the conditional probability an individual benefits from treatment in pre-treatment covariate stratum $x$. To understand the effect of treatment for all individuals regardless of strata, denote
\[
\theta(\delta) :={\rm pr}\left\{Y_i(1)-Y_i(0)>\delta\right\}
\]
as the marginal probability an individual benefits from treatment.
\end{definition}

Here, $X_i$ is a pre-treatment covariate that is thought to deconfound variability in the outcome $Y_i$ that is not due to $W_i$ and exogenous noise alone \citep{Rubin1974EstimatingCE,imbens_rubin_2015}, as we state formally in Assumption \ref{assump:strongCondIgno}. One can see that $\theta(\delta)= \E{\prc{ Y_i(1)-Y_i(0)>\delta }{X_i}}$, where the expectation is taken with respect to the confounder $X_i$. Consider now two very similar quantities in Definition \ref{defn:mannU}. 

\begin{definition}\label{defn:mannU}

To understand heterogeneity in treatment's effect across differing individuals and across differing strata of a pre-treatment covariate, denote
\[
\eta(\delta,x) := \mathrm{pr}\{Y_i(1)-Y_j(0)>\delta\mid X_i=X_j=x;i\neq j\}
\]
as the probability that a randomly selected individual in pre-treatment covariate stratum $x$ has a treatment potential outcome that is better than the control potential outcome of a differing randomly selected individual also in stratum $x$. To understand treatment's effect across differing individuals overall, denote
\[
\eta(\delta) :=\mathrm{pr}\{Y_i(1)-Y_j(0)>\delta;i\neq j\}
\]
as the overall probability that a randomly selected individual has a treatment outcome that is better than a differing randomly selected individual's control outcome.
\end{definition}

Of crucial importance, $\prc{Y_i(1)-Y_i(0)>\delta}{X_i}$ and $\theta(\delta)$ are not in general the same as $\eta(\delta,x)$ and $\eta(\delta)$ in Definition \ref{defn:mannU} \citep{hand1992comparing,fayMannU2018,Greenland2020-tu}. This is because $Y_i(1)$ and $Y_i(0)$ are in general {dependent} random variables, while $Y_i(1)$ is independent of $Y_j(0)$ when $i\neq j$ under the standard Stable Unit Treatment Value Assumption (Assumption \ref{assump:sutva}). Under appropriate identifiability assumptions, such as Assumption \ref{assump:strongCondIgno} below, $\eta(\delta,x)$ and $\eta(\delta)$ can be identified because we can generally sample from the conditional distribution of $Y_i$ given  $(W_i,X_i)=(w,x)$ and the conditional distribution of $Y_i$ given $W_i=w$ $(w=0,1)$. It is impossible to sample from the joint distributions of $(Y_i(0),Y_i(1))$ marginally or given $X_i=x$, because an individual cannot be in both the treatment group and the control group simultaneously. So we cannot generally identify $\pibt$. Not even if we are able to perfectly match individuals in opposite treatment groups based on pre-treatment covariates. This is what is known as the fundamental problem of causal inference.

The goal nonetheless is to reason about $\pibt$ through estimators for bounds on this quantity. Our focus is on deriving closed-form, non-asymptotic margins of error $(\varepsilon)$ on the bound estimators for an overall confidence band on $\pibt$ of the form: with a desired frequentist confidence level, $\pibt$ is contained between  
\begin{equation}\label{eqn:interpretWordsL}
\parenth{\text{estimator for lower bound on $\pibt$}}-\varepsilon\text{ and }\parenth{\text{estimator for upper bound on $\pibt$}}+\varepsilon.
\end{equation}


\subsection{Overview of our contributions}

For the bounds on the marginal probability an individual benefits from treatment which are estimated with data from a randomized experiment or observational study with known treatment probabilities, we derive in Section \ref{sec:margITEMakBnds} a closed-form concentration inequality depending on only the sample size and the desired frequentist confidence level. As discussed in Section \ref{sec:powerAnalysisMargPIBTCase}, this allows for a formal statistical power analysis, albeit conservative, but notably without the requirement of an asymptotic limiting distribution nor the specification of any unknown parameters (e.g. plausible effect sizes). Our main results are presented in a general manner in terms of sub-groups, delineated by pre-treatment features, and estimators for bounds on $\pibt$ based on inverse probability of treatment weighting (IPTW) \citep{imbens_rubin_2015,hernan2020}. The inference on $\pibt$ with a randomized experiment is handled as a specific case.

Different from the non-asymptotic margin of error that can be obtained with bootstrap re-sampling \citep{Efron1994-oc,bickelEtAl1997}, our non-asymptotic margin of error will be closed-form and simultaneous for all thresholds $\delta$ that can be used to define $\pibt$, thus allowing for a form of sensitivity analysis on its definition, and even cherry-picking. Also different from a margin of error that can be obtained via Bootstrap, our approach can be up to an order $M$ (the number of bootstrap samples) faster; see Remark \ref{rem:compComplex} for details on the relatively low computational complexity of our approach. 

We include in Section \ref{sec:criteoUpliftExample} an example application to a real-life randomized experiment dataset, Criteo AI Lab's benchmark data for uplift prediction \citep{CriteoUpliftDiemert2018}. In particular, this section points toward a useful combination of conditional average treatment effect ($\cate$) estimation and the inference we provide for $\pibt$ in this work: through a partitioning of individuals in a sample based on their similar $\cate$ prediction, we can estimate bounds on $\pibt$ in each of these strata to better understand the implication of the $\cate$ estimate. This is related to work on causal analysis involving the stratification on an interpretable score, such as a prognostic score, propensity score, or other diagnostic score \citep{Abadie_endogStrat2018,padilla2021causal,yeChenPadilla2021Scores,YadlowskyWagerTrtmntPrioritizationRules2021}. It is also related to recent efforts to provide valid and benchmarked causal inferences \citep{validatingCauslInfMethods22,calibrationHeteroTEs_XuYadlowsky2022,benchmarkHetTrtEffInterpretability2022,validInferenceAfterDiscovery}.


For the case that we would like to consider heterogeneity in $\pibt$ beyond the sub-group level as we consider in Section \ref{sec:margITEMakBnds}, Proposition \ref{prop:CIsGaussResids} provides an application of general results in Appendix \ref{sec:condITEMakBnds} (see extended discussion there) to the canonical linear regression model. Moreover, the Appendix \ref{sec:whenWorks} is an extended discussion on the scope of our results. For binary potential outcomes, we show in Proposition \ref{prop:boolFrechVsMakarov} that the general approach we take to bound $\pibt$ at the population-level using \citet{makarovBounds} is equivalent to using the sharp Boole-Fr\'echet bounds \citep{FrechetMaurice1935Gdtd,Rschendorf1981SharpnessOfFrechetBoole}. Our results also extend to the case when we define benefiting from treatment in terms of the ratio of an individual's two positive potential outcomes as discussed in Section \ref{sec:diffDefnRatio}. This section also discusses how our results can easily extend to reasoning about the proportion of individuals that are harmed by an intervention \citep{kallus2022_propHarmed}.

Appendix \ref{append:proofsMainResults} contains all the proofs of the main results presented here. 
\subsubsection{Existing work: bounding the distribution of individual treatment effects\label{sec:boundingITEDistr}}

The contribution of the present work, relative to the contribution of \citet{fanSooPark2010} who discuss inference for only the randomized experiment setting, is the concentration inequality for the $\pibt$ bound estimators. Under regularity conditions, \citet{fanSooPark2010} show asymptotically that the plug-in bound estimators follow either a normal distribution (centered at the target bound), or a truncated normal distribution, or a point mass. Exactly which distribution this is depends on the supremum difference between the two potential outcomes' cumulative distribution functions ($\cdf$s), which is unkown. Even if we know that the asymptotic distribution of the estimator is Gaussian, a prospective power analysis further requires an estimator for the standard error to guarantee a target confidence level and margin of error (e.g. a maximum deviation of 0.05), but such an estimator is not discussed in \citet{fanSooPark2010}. This points to a strength of our main concentration result in the randomized experiment setting: despite the possibility that the plug-in estimator can have a non-trivial, possibly biased, sampling distribution in a finite sample, the confidence level we can have for a target margin of error depends only on sample size. The discussion around Proposition \ref{prop:mainResultMargPIBT} and Fig. \ref{fig:powerAnalysisMargPIBTCase} gives more details. 

\citet{fayMannU2018} also discuss the statistical inference technique of \citet{fanSooPark2010} in conjunction with the quantity $1-\eta(\delta)$ in Definition \ref{defn:mannU}. Interestingly, it has been established that the Makarov bounds for the marginal $\cdf$ of $Y_i(1)-Y_i(0)$ studied in \citet{fanSooPark2010} are point-wise but not uniformly sharp \citep{firpo2010,FIRPO2019210}. {} While promising, we consider the estimation of these tightened bounds beyond the scope of this paper as it is not immediately clear that it is amenable to our analysis. For continuous outcomes in a randomized experiment setting, \citet{frandsenPositiveCor2021} works under a condition known as mutual stochastic increasing-ness of the potential outcomes $(Y_i(0),Y_i(1))$ \citep{Lehman1966}.{}{} The plug-in estimation approach we use in the randomized experiment case does not make the assumption of positive correlation: it works for any joint distribution on $(Y_i(0),Y_i(1))$ \citep{fanSooPark2010}, including those with any type of negative association.{} Also in the context of a randomized experiment, \citet{Caughey2021RandomizationIB} study $\pibt$ under a randomization inference setup that is traditionally used to test the sharp null hypothesis that all individual treatment effects are constant \citep{Fisher1935}.{} The approach we take to bound $\pibt$ assumes the existence of an infinite super-population that subjects in our sample at hand are drawn independent and identically distributed from and for which our plug-in estimators provide inference for. \citet{Caughey2021RandomizationIB} appears to be a nice alternative under the differing assumption that randomness is solely due to random assignment of subjects to a treatment.

Of special note, the quantity $\theta(\delta)$ in Definition \ref{defn:propBenefit} when $0\leq\delta<1$ is equivalent to what is known as the ``probability of necessity and sufficiency ($\pns$)'' if $Y_i(0)$ and $Y_i(1)$ only take on binary values \citep{Pearl1999-sh,Tian2000-ul,pearlCausality,muellerLiPearl2021PNSCovariates}. In this case, $\pns$ and what we call ``marginal $\pibt$'' are given by the joint probability $\mathrm{pr}\{Y_i(1)=1,Y_i(0)=0\}$. As suggested by the intriguing use of prepositional logic terminology in its name, $\pns$ informs us of an intervention's effectiveness at achieving a strictly better outcome.{} See also \citet{Cuellar2018-tk}, \citet{Cinelli2021-id}, \citet{dawidMusio2022}, and \citet{Cuellar2022} for recent discussions on related quantities, meticulously defined as conditional probabilities of one potential outcome given the other, that inform about treatment's efficacy. 

The use of IPTW to bound PNS when it is defined for inferring harm from an intervention is considered in \citet{kallus2022_propHarmed}. Relative to their handling of PNS (binary $\pibt$), our results are for both real-valued and binary-valued outcomes. We discuss the sharpness of our target bounds for the binary case in Appendix \ref{sec:whenWorks}. Moreover, our non-asymptotic inference is complementary to the invocation of the Central Limit Theorem in \citet{kallus2022_propHarmed}.

Our work is also complementary to inference on the inverse of $\pibt$ \citep{fanSooPark2012,conformalITE2021,YinShiWangBlei2021SensitivityITE,JinRenCandes2021SensitivityITE}, such as with quantile regression \citep{regressionQuantilesKoenkerBasset1978,romanoConformal2019}.


{}

\subsection{Assumptions}

We now state the conditions we will work with. Throughout this work, we will assume the stable unit treatment value assumption, commonly abbreviated as SUTVA, in Assumption \ref{assump:sutva}. We will also assume consistency of the observed outcome throughout as given in Assumption \ref{assump:consistency}. Moreover, our inference is valid under conditional ignorability conditions in Assumption \ref{assump:strongCondIgno}.

\begin{assumption}[Stable Unit Treatment Value Assumption\label{assump:sutva}]

We quote \citet{imbens_rubin_2015}: ``The potential outcomes for any unit do not vary with the treatments assigned to other units, and, for each unit, there are no different forms or versions of each treatment level, which lead to different potential outcomes.''
\end{assumption}

\begin{assumption}[Consistency\label{assump:consistency}]

The observed outcome is dictated by treatment receipt indicator $W_i$:
\[
Y_i:=W_i Y_i(1)+(1-W_i)Y_i(0).
\]
\end{assumption}






\begin{assumption}[Strong Conditional Ignorability\label{assump:strongCondIgno} \citep{Rubin1974EstimatingCE,imbens_rubin_2015}]

For $i=1,\dots,n$, we have that:
\[
\curl{Y_i(0),Y_i(1)}\indep W_i\mid X_i.
\]
Let $e(x):=\prc{W_i=1}{X_i=x}$--the propensity score. We have further that 
\[
0 < \underbar{e}:=\inf_x e(x)\text{ and }\bar{e}:=\sup_x e(x) < 1.
\]
\end{assumption}

\section{$\pibt$ bounds in sub-groups\label{sec:margITEMakBnds}}
\subsection{Overview}
Here, we aim to estimate bounds on \[
\theta(\delta,x):=\prc{Y_i(1)-Y_i(0)>\delta}{\ell(X_i)=\ell(x) }.
\]
Here, $\ell(x)=l$ if $x\in\xcal_l$, where
\[
\xcal_l\subseteq \xcal:=\text{support}(X_i)\ (l=1,\dots,L;\ L\geq 1).
\]
We consider the mapping $\ell(\cdot)$ to be fixed. However, among other possibilities, $\ell(\cdot)$ could be the partition given by a decision tree \citep{Breiman2017-ay,grfAnnals,policytree2021,poliTreeSofware} or uniform mass binning \citep{unifMassBinning2001,Abadie_endogStrat2018,verifiedUncerCalib2019,distrFreeBinaryReg2020,pmlr-v139-gupta21b} during a pre-processing step with independent data. See \S~\ref{sec:criteoUpliftExample} below for an example. We may also trivially take $L=1$ as a special case, like we do in Figure \ref{fig:powerAnalysisMargPIBTCase}.

In order for us to identify these bounds, we will work under Assumption \ref{assump:strongCondIgno}. Denote \citet{makarovBounds}'s pointwise sharp lower bound and upper bound on $\pibt$ as $\theta^L(\delta,x)$ and $\theta^U(\delta,x)$, which are such that
\[
\theta^L(\delta,x)\ \leq\ \theta(\delta,x)\ \leq\ \theta^U(\delta,x).
\]
Denote their corresponding estimators based on independent and identically distributed data as $\hat{\theta}^L(\delta,x)$ and $\hat{\theta}^U(\delta,x)$, respectively. Proposition \ref{prop:mainResultMargPIBT} is the main result in this section, providing a guarantee for the accuracy of these estimators, which we now formally define. Remark \ref{rem:compComplex} discusses the satisfactory computational complexity of our procedure, while Section \ref{sec:powerAnalysisMargPIBTCase} discusses how to use Proposition \ref{prop:mainResultMargPIBT} for sample size deterimination (power analysis). Moreover, Section \ref{sec:margCaseSimulations} of the appendix contains a simulation study.

\subsection{The target bounds on $\pibt$ and their estimators}

We refer the reader to Lemma \ref{lem:makBounds} in Appendix \ref{append:makBoundsDiscussion} for the formal statement of the point-wise sharp Makarov bounds \citep{makarovBounds,WILLIAMSON199089} we make use of to bound $\pibt$. The bounds $\theta^L(\delta)$ and $\theta^U(\delta)$, our target parameters, are in terms of the potential outcomes' marginal $\cdf$s. The $\cdf$ of $Y_i(w)$ in sub-group $\ell(x)$ is: 
\[
F_w(y|x):=\prc{Y_i(w)\leq y}{\ell(X_i)=\ell(x) },
\]
which is identified under Assumption \ref{assump:strongCondIgno}. Denote the empirical cumulative distribution function ($\ecdf$) for $Y_i(w)$, an unbiased plug-in estimator for $F_w(y|x)$, as $(w=0,1)$:
\begin{equation}\label{eqn:eCDFmarg}
\hat{F}_{wn}(y|x):=\frac{1}{n_{\ell(x )}}\sum_{i:\ \ell(X_i)=\ell(x)}\frac{\indic{Y_i\leq y} \indic{W_i=w} }{ we(X_i)+(1-w)[1-e(X_i)] }.
\end{equation}
Here, $\indic{\cdot}$ is the indicator function, while
\[
n_{\ell(x )}:=\#\curl{i:\ \ell(X_i)=\ell(x) },
\]
the number of units in sub-group $\ell(x)$. Unbiasedness of $\hat{F}_w(y|x)$ follows by noting that $\hat{F}_w(y|x)$ is an inverse probability of treatment weighted (IPTW) estimator \citep{imbens_rubin_2015,hernan2020} for the conditional mean of an indicator:
\[
F_w(y|x)=\EC{\indic{Y_i(w)\leq y}}{\ellxi = \ellx }.
\]

Using Lemma \ref{lem:makBounds} in the Appendix \ref{append:makBoundsDiscussion}, the target parameters to bound $\theta(\delta,x)$ across any joint distribution of $\curl{Y_i(0),Y_i(1)}\mid\ellxi=\ell(x)$ are:
\[\begin{aligned}
&{\theta}^{L}(\delta,x)&=\ &-\min\left[\inf_y\left\{F_1(y+\delta/2\mid x)-F_0(y-\delta/2\mid x)\right\},0\right]\\
\end{aligned}\]
for the lower bound, while for the upper bound we have:
\[\begin{aligned}
&{\theta}^{U}(\delta,x)&=\ &1-\max\left[\sup_y\left\{ F_1(y+\delta/2\mid x)-F_0(y-\delta/2\mid x) \right\},0\right].\\
\end{aligned}\]
Correspondingly, we can obtain the bound estimators by plugging in the $\cdf$ estimators in analogy to \citet{fanSooPark2010} who consider only the randomized experiment case: 
\begin{equation}\label{eqn:lwrEstRE}
\hat{\theta}^{L}(\delta,\mid x):=-\min\left[\inf_y\curl{\hat{F}_{1n}(y+\delta/2\mid x)-\hat{F}_{0n}(y-\delta/2\mid x)},0\right]
\end{equation}
and
\begin{equation}\label{eqn:uprEstRE}
\hat{\theta}^{U}(\delta,x):=1-\max\left[\sup_y\curl{\hat{F}_{1n}(y+\delta/2\mid x)-\hat{F}_{0n}(y-\delta/2\mid x)},0\right]
\end{equation}
are the lower bound and upper bound estimators, respectively. 

\begin{remark}[Computational complexity\label{rem:compComplex}]
For the computation of the lower and upper bound estimates at a point $x$, the number of total iterations across elements in two lists of size $n_{0\ell(x)}$ and $n_{1\ell(x)}$, respectively, is
\[
O[\parenth{n_{0\ell(x)}+n_{1\ell(x)} }\{\log \parenth{n_{0\ell(x)}}+\log \parenth{n_{1\ell(x)}} \}].
\]
See Appendix \ref{append:compComplex} for details. Here, $n_{w\ell(x)}$ $(w=0,1)$ is the number of units in treatment group $w$ and sub-group $\ell(x)$.

\end{remark}


\subsection{The main result}
Given our estimators $\{\hat{\theta}^L(\delta,x),\hat{\theta}^U(\delta,x)\}$, we consider now their accuracy. Proposition \ref{prop:mainResultMargPIBT} provides us with this understanding. Here, the symbol $\vee$ denotes the maximum between the left and right arguments. 

\begin{proposition}[Inference on $\pibt$ through bounds\label{prop:mainResultMargPIBT}]\ \\
Let Assumption \ref{assump:strongCondIgno} hold.\\
\textbf{Case 1:} Consider evaluating the bound estimators at a random new test point, $\xnp$, independent of our training data. 

With probability at least $1-\alpha$, we have:
\[\begin{aligned}
& \sup_\delta \curl{\abs{ \hat{\theta}^L(\delta,\xnp)-\theta^L(\delta,\xnp) }\vee\abs{ \hat{\theta}^U(\delta,\xnp)-\theta^U(\delta,\xnp) }}\\
\leq\ &n_{\ell(\xnp )}^{-1/2}\sum_{w=0,1} C_{w,\alpha/2}(\xnp).\end{aligned}\]

\textbf{Case 2:} Now consider a uniform guarantee on the bound estimators across $x\in\xcal$.

Let $x_l\in\xcal_l$ be arbitrary ($l=1,\dots,L$). With probability at least $1-\alpha$, we have:
\[\begin{aligned}
&\sup_{\delta,x} \curl{\abs{ \hat{\theta}^L(\delta,x)-\theta^L(\delta,x) }\vee\abs{ \hat{\theta}^U(\delta,x)-\theta^U(\delta,x) }}\\
\leq\ &\max_l \curl{ n_l^{-1/2}\sum_{w=0,1}  C_{w,\alpha/(2L)}(x_l)}.
\end{aligned}\]

\end{proposition}

In Proposition \ref{prop:mainResultMargPIBT}, we define $C_{w,\beta}$, where $\beta=\alpha/2$ (Case 1) or $\beta=\alpha/(2L)$ (Case 2), as follows.
\begin{itemize}
\item\textbf{Completely Randomized Experiment in sub-group $\ell(x)$}: Suppose $\ell(x)=l$. If $e(x')=n_{1l}/n_l$ for each $x'\in \xcal_l$, we have:
\[
C_{w,\beta}(x)\ \dot{=}\ \parenth{2^{-1}n_{\ell(x )} n_{w\ell(x)}^{-1}\log(2/\beta)}^{1/2}.
\]
\item\textbf{General case with known propensity scores}: We have:

\[
C_{w,\beta}(x)\ \dot{=}\ \frac{2\sqbrack{ 2\log^{\frac{1}{2}}\parenth{ n_{\ell(x)}+1 }+\log^{\frac{1}{2}}(1/\beta)}}{ n_{\ellx} ^{1/2}\parenth{ w\underbar{e}+(1-w)(1-\bar{e}) } }.
\]
\end{itemize}
The exact choice of these values is explained in Lemma \ref{lem:eCDFGuarantee}, a key result that makes use of the Dvoretzky-Kiefer-Wolfowitz-Massart (DKWM) inequality \citep{DKWIneq1956,massartDKWIneq1990,NaamanDKWIneq2021,overlap2021}, a functional Hoeffding's theorem \citep{wainwright_2019}, and techniques from empirical process theory \citep{wainwright_2019,vanderVaart1996}.

In Remark \ref{rem:practicalInterpretation}, we can see one practical implication of Proposition \ref{prop:mainResultMargPIBT}. In Section \ref{sec:powerAnalysisMargPIBTCase}, we further discuss its practical implication with respect to a statistical power analysis. In \S~\ref{sec:criteoUpliftExample}, we apply Case 2 of Proposition \ref{prop:mainResultMargPIBT}.

\begin{remark}\label{rem:practicalInterpretation}
Using Proposition \ref{prop:mainResultMargPIBT}, we can say that with confidence at least $(1-\alpha)\times100\%$, the probability an individual represented by our study will benefit from treatment,
\[
\prc{ Y_i(1)-Y_i(0)>\delta }{ \ell(\xnp) }
\]
for any threshold $\delta$ of interest, is between
\[
\hat{\theta}^L(\delta,\xnp)-n_{\ell(\xnp) }^{-1/2}\sum_{w=0,1} C_{w,\alpha/2}(\xnp)
\]
and
\[
\hat{\theta}^U(\delta,\xnp)+n_{\ell(\xnp) }^{-1/2}\sum_{w=0,1} C_{w,\alpha/2}(\xnp).
\]

\end{remark}

\subsection{A power analysis with Proposition \ref{prop:mainResultMargPIBT}\label{sec:powerAnalysisMargPIBTCase}}

\begin{figure}[ht]
    \centering
    \includegraphics[width=0.45\textwidth]{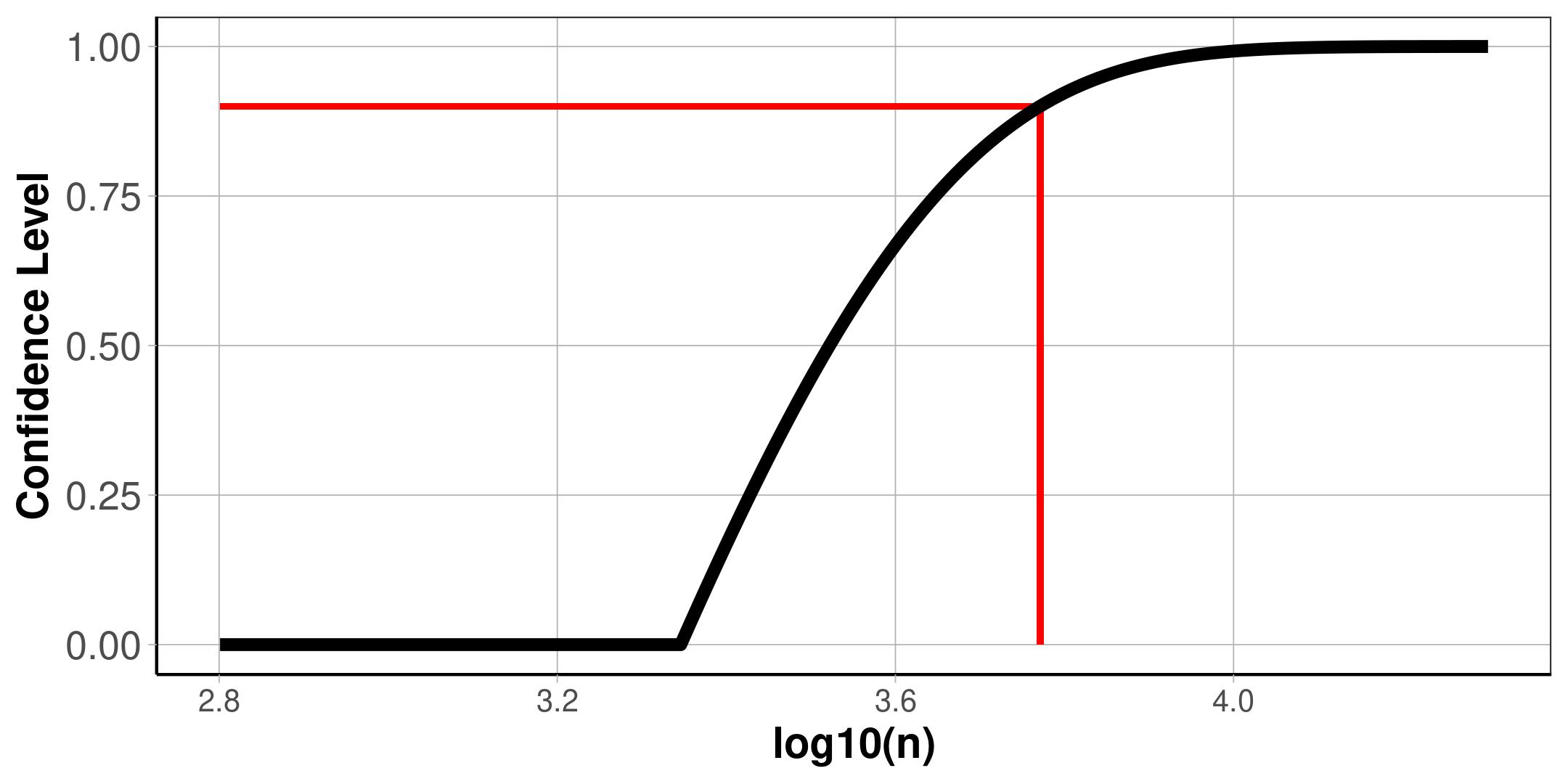}
    \caption{An application of Proposition \ref{prop:mainResultMargPIBT}. Here, $L=1$, $n_0=n_1$, and  $n=n_0+n_1$. The target margin of error is $0.05$. The black curve, given by \eqref{eqn:sig_level_plot}, corresponds to the possible confidence level we may have at a given $n$. From the two red line segments, we see that a sample size of $n\geq 5902$ guarantees the target margin of error with at least a 90\% confidence level.}
    \label{fig:powerAnalysisMargPIBTCase}
\end{figure}

Figure \ref{fig:powerAnalysisMargPIBTCase} presents an example power analysis making use of Proposition \ref{prop:mainResultMargPIBT} for the completely randomized experiment case. More generally, suppose our target margin of error for
\[
\sup_{\delta}\curl{\abs{ \hat{\theta}^{L}(\delta,x)-\theta^{L}(\delta,x) }\vee\abs{ \hat{\theta}^{U}(\delta,x)-\theta^{U}(\delta,x) }}
\]
in the completely randomized experiment case is $0<\varepsilon<1$. Solving for the significance level $\alpha_\varepsilon$ when we set $\varepsilon$ equal to the margin of error in Proposition \ref{prop:mainResultMargPIBT}:
\[
\varepsilon\ \dot{=}\ \curl{\frac{\log(4/\alpha_{\varepsilon} )}{2} }^{1/2}\parenth{n_{0\ell(x)}^{-\frac{1}{2}}+n_{1\ell(x)}^{-\frac{1}{2}} }
\]
gives:
\begin{equation}\label{eqn:sig_level_plot}
\alpha_{\varepsilon} = 4\exp\curl{ -2 \parenth{n_{0\ell(x)}^{-\frac{1}{2}}+n_{1\ell(x)}^{-\frac{1}{2}} }^{-2} \varepsilon^2 }.
\end{equation}
The confidence level we can thus have for the target margin of error $\varepsilon$ at any given sample size $(n_{0\ell(x)},n_{1\ell(x)} )$ is at least:
\[
\begin{cases}
1-\alpha_{\varepsilon}&\text{ if }0\leq\alpha_{\varepsilon}\leq 1\\
0&\text{ otherwise.}
\end{cases}
\]

\section{Application: Understanding the implication of conditional average treatment effects\label{sec:criteoUpliftExample}}
\subsection{One practical definition of sub-group level $\pibt$ and its bounds}
Using pre-treatment covariates, $X_i$, we will study heterogeneity and the implication of an average treatment effect estimate as follows. We first learn the conditional average treatment effect ($\cate$) function, 
\[
\tau(x):=\mathbb{E}\curl{Y_i(1)-Y_i(0)\mid X_i=x}.
\]

We then partition the feature space based on similar $\cate$ predictions through uniform mass binning\citep{unifMassBinning2001,Abadie_endogStrat2018,verifiedUncerCalib2019,distrFreeBinaryReg2020,pmlr-v139-gupta21b} as follows. Denote the quantiles of $\hat{\tau}(X_i)$ as 
\[
q_\alpha:=\inf\curl{q:\ {\rm pr}\curl{\hat{\tau}(X_i)\leq q}\geq\alpha}\ (0\leq\alpha\leq1).
\]
We will obtain estimators, $\hat{q}_\alpha$ using a second fold of randomly sampled rows. Now, define the discrete mapping $s_m:\ \mathcal{X}\to\{1,\dots,m\}$ as:
\[
s_m(x) = k\text{ if }\hat{q}_{(k-1)/m}<\hat{\tau}(x)\leq\hat{q}_{k/m};\ k\in\{1,\dots,m\}.
\]
The value $m\geq2$ corresponds to the number of sub-groups that are created in the partitioning. 

Finally, we estimate bounds on $\pibt$ conditional on the partition, defined as
\begin{equation}\label{eqn:trtmntBeneficStrat}
\theta(\delta,k):={\rm pr}\left\{Y_i(1)-Y_i(0)>\delta \mid s_m(X_i)=k\right\}.
\end{equation}
The parameter $\theta(\delta,k)$ is the probability that treatment is beneficial in $\cate$ stratum $k=1,\dots,m$. When treatment is randomized independently of the baseline covariates $X_i$, as we assume in this section, we have the ignorability statement \[\curl{Y_i(1),Y_i(0)}\indep W_i\mid s_m(X_i)=k\ (k=1,\dots,m).\] This allows us to identify the Makarov lower and upper bounds on $\theta(\delta,k)$, which we will denote as $\theta^L(\delta,k)$ and $\theta^U(\delta,k)$, respectively. The mappings $s_m(\cdot)$ allow us to stratify on an interpretable univariate score. For our demonstrative choice of the mapping $s_m(\cdot)$, if the learned function CATE function $\hat{\tau}(\cdot)$ is indicative of benefiting from treatment, we would like to see a monotone increasing relation between $k=1,\dots,m$ and the bound estimators $\{\hat{\theta}^L(\delta,k),\hat{\theta}^U(\delta,k)\}$ as supporting evidence. 


\subsubsection{Criteo's uplift prediction benchmark dataset}


\begin{figure}[ht]
    \centering
    \includegraphics[width=0.45\textwidth]{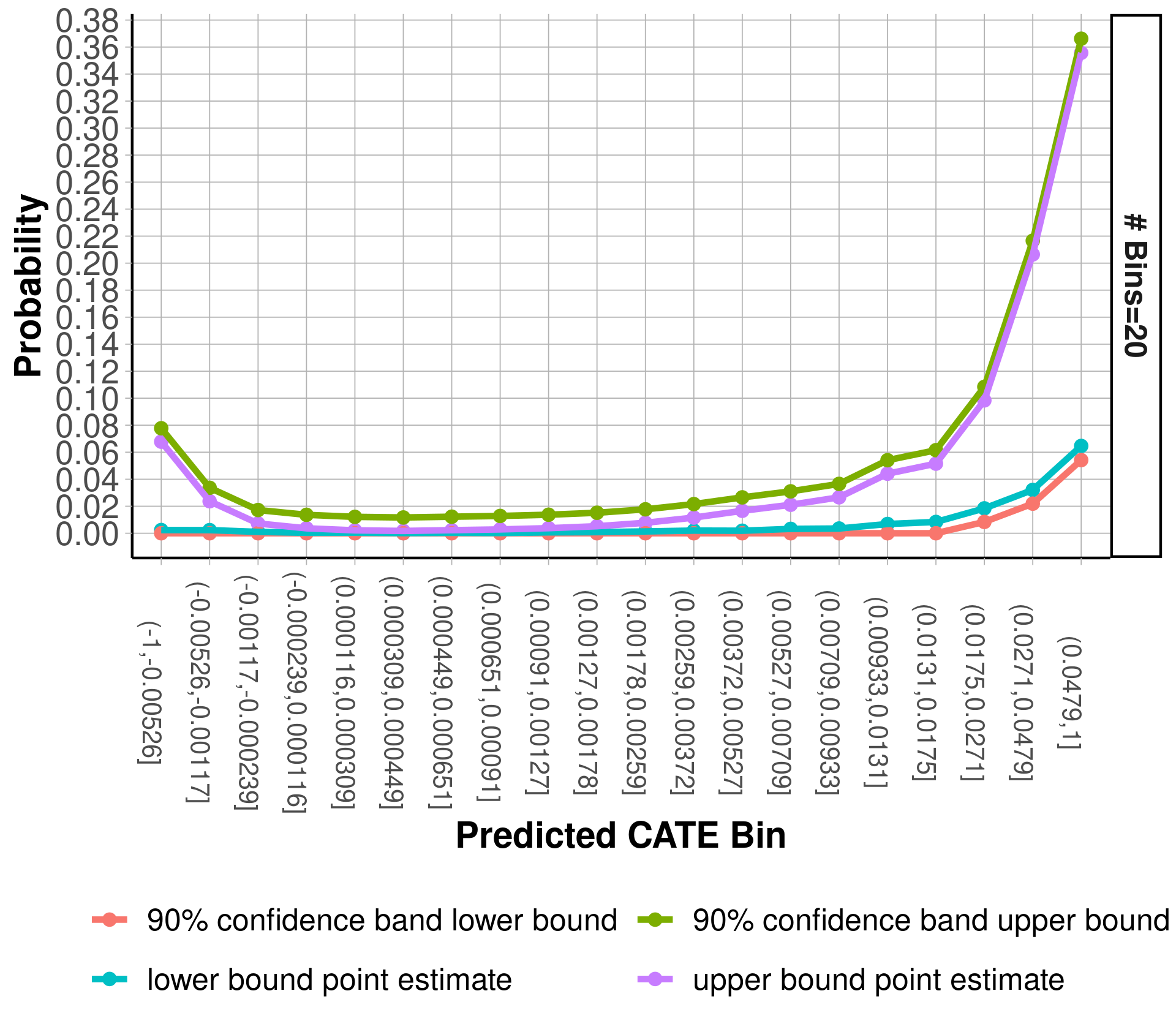}
    \caption{The 90\% Bonferroni corrected lower and upper confidence bands on $\pibt$ across bins of $\cate$ predictions on the Criteo uplift dataset.}
    \label{fig:criteoByBin}
\end{figure}

We now present an application to Criteo AI Lab's uplift prediction benchmark dataset \citep{CriteoUpliftDiemert2018}. According to the webpage that hosts the data,
\begin{quote}
This dataset is constructed by assembling data resulting from several incrementality tests, a particular randomized trial procedure where a random part of the population is prevented from being targeted by advertising. It consists of 25M rows, each one representing a user with [12] features, a treatment indicator and 2 labels (visits and conversions).
\end{quote}
For this application, the proportion we will estimate bounds for should be understood in plain language as the proportion of the time that the advertiser benefits from presenting an advertisement on the website, rather than the probability an individual benefits from treatment. Equation \eqref{eqn:trtmntBeneficStrat} below is the formal statement of this proportion. 

The available down sampled data consists of 13,979,592 observations. We focus on the effect treatment assignment, rather than treatment receipt, has on visits, making this an intent-to-treat analysis. The outcome of interest in our analysis is the indicator for whether a user visited the advertiser website during the test period of 2 weeks. 

With 50,000 randomly sampled rows, we learn $\cate$ through \texttt{grf::causal\_forest} with default options in \texttt{R} \citep{grfAnnals}. Next, we obtain estimators for the $m-1$ quantiles of the random variable $\hat{\tau}(X_i)$ using a second fold of 8,987,000 rows. The bounds are then estimated within each bin in the second fold.

Restricting $\delta\in(0,1)$ and noting that the outcome is binary, Figure \ref{fig:criteoByBin} provides a 90\% Bonferroni corrected confidence band on 
\[
\theta(\delta,k)={\rm pr}\curl{Y_i(1)=1,Y_i(0)=0\mid s_m(X_i)=k}
\]
that allows for simultaneous inference across $\cate$ prediction bins $k=1,\dots,20$. The practical insight is that for an individual such that $0.0479<\hat{\tau}(X_i)\leq 1$, the probability this individual will visit the advertiser's web-page when assigned treatment and otherwise not visit the website if untreated, is between 5.42\% and 36.62\% with 90\% confidence. Given the Bonferroni correction and that all bins have equal amounts of subjects, we take an average of the upper confidence bound for every bin corresponding to a $\cate$ prediction of $0.0271$ or less. This gives that an individual with a $\cate$ prediction of $0.0271$ or less has a joint probability of interest of 4.18\% or less with 90\% confidence. 



\section{Discussion}

Promising directions related to the work presented here include applying the results in Appendix \ref{sec:condITEMakBnds} to a setting more general than that demonstrated in Proposition \ref{prop:CIsGaussResids}. Alternatively, we may like to tailor these results to model specifications that are sufficient for interesting applications, such as involving generalized linear models \citep{McCullagh2019-qz,modernHighDLogisRegr_SurCandes2019}. For settings where comparing an individual's two potential outcomes in a ratio rather than a difference can provide interesting insight, such as studies where time to an event is of interest \citep{coxRegressionHazardRatio_1972,Stitelman2010_TTE_TMLE,Austin2014_TTE,Schober2018_TTE,Cai2020_TTE_TMLE}, it seems worthwhile to extend the discussion in Appendix \ref{sec:diffDefnRatio}.

Our probability bound formulation to reason about treatment effects, rather than the more common average formulation, is similar to recent model-free work that moves beyond an average in favor of controlling the probability of a type I error, or a false discovery, using a Neyman-Pearson paradigm \citep{TongFengLi2018_NPClassification,LiChenTong2021_NPFeatureRanking}. Future work relating our present work to these advances could be of interest, for example as it relates to causal classification \citep{causalClassification2022} and mitigating harm from an intervention \citep{kallus2022_propHarmed}. This direction may involve more rigorously choosing the appropriate threshold $\delta$ to define $\pibt$ in Definition \ref{defn:propBenefit}, in accordance with an appropriately formulated Neyman-Pearson objective. Recalling our uniform control across threshold values, we currently propose that $\delta$ be anything, or everything, a practitioner finds reasonable.

\section*{Acknowledgement}
Gabriel would like to thank NSF DGE-1650604 for financial support, Qing Zhou for encouragement on this project, and members of the UCLA causal inference reading group for their feedback.

\bibliography{references}

\begin{thebibliography}{}

\bibitem[Abadie et~al., 2018]{Abadie_endogStrat2018}
Abadie, A., Chingos, M.~M., and West, M.~R. (2018).
\newblock {Endogenous Stratification in Randomized Experiments}.
\newblock {\em The Review of Economics and Statistics}, 100(4):567--580.

\bibitem[Athey et~al., 2019]{grfAnnals}
Athey, S., Tibshirani, J., and Wager, S. (2019).
\newblock {Generalized random forests}.
\newblock {\em The Annals of Statistics}, 47(2):1148 -- 1178.

\bibitem[Athey and Wager, 2021]{policytree2021}
Athey, S. and Wager, S. (2021).
\newblock {Policy Learning With Observational Data}.
\newblock {\em Econometrica}, 89(1):133--161.

\bibitem[Austin, 2014]{Austin2014_TTE}
Austin, P.~C. (2014).
\newblock The use of propensity score methods with survival or time-to-event
  outcomes: reporting measures of effect similar to those used in randomized
  experiments.
\newblock {\em Stat. Med.}, 33(7):1242--1258.

\bibitem[Bickel et~al., 1997]{bickelEtAl1997}
Bickel, P.~J., Götze, F., and van Zwet, W.~R. (1997).
\newblock Resampling fewer than n observations: Gains, losses, and remedies for
  losses.
\newblock {\em Statistica Sinica}, 7(1):1--31.

\bibitem[Boole, 1854]{boole1854}
Boole, G. (1854).
\newblock {\em An investigation of the laws of thought : on which are founded
  the mathematical theories of logic and probabilities / By George Boole}.
\newblock Walton and Maberly, London.

\bibitem[Breiman et~al., 2017]{Breiman2017-ay}
Breiman, L., Friedman, J.~H., Olshen, R.~A., and Stone, C.~J. (2017).
\newblock {\em Classification And Regression Trees}.
\newblock Routledge.

\bibitem[Burkhart and Ruiz, 2022]{BurkhartRuiz2022}
Burkhart, M.~C. and Ruiz, G. (2022).
\newblock Neuroevolutionary feature representations for causal inference.
\newblock In Groen, D., de~Mulatier, C., Paszynski, M., Krzhizhanovskaya,
  V.~V., Dongarra, J.~J., and Sloot, P. M.~A., editors, {\em Computational
  Science -- ICCS 2022}, pages 3--10, Cham. Springer International Publishing.

\bibitem[Cai and van~der Laan, 2020]{Cai2020_TTE_TMLE}
Cai, W. and van~der Laan, M.~J. (2020).
\newblock One-step targeted maximum likelihood estimation for time-to-event
  outcomes.
\newblock {\em Biometrics}, 76(3):722--733.

\bibitem[Caughey et~al., 2021]{Caughey2021RandomizationIB}
Caughey, D., Dafoe, A., Li, X., and Miratrix, L. (2021).
\newblock Randomization inference beyond the sharp null: Bounded null
  hypotheses and quantiles of individual treatment effects.

\bibitem[Chernozhukov et~al., 2013]{chernozhukovDiazDistrTrtmntEff2013}
Chernozhukov, V., Fernández-Val, I., and Melly, B. (2013).
\newblock Inference on counterfactual distributions.
\newblock {\em Econometrica}, 81(6):2205--2268.

\bibitem[Cinelli and Pearl, 2021]{Cinelli2021-id}
Cinelli, C. and Pearl, J. (2021).
\newblock Generalizing experimental results by leveraging knowledge of
  mechanisms.
\newblock {\em Eur. J. Epidemiol.}, 36(2):149--164.

\bibitem[Cox, 1972]{coxRegressionHazardRatio_1972}
Cox, D.~R. (1972).
\newblock Regression models and life-tables.
\newblock {\em Journal of the Royal Statistical Society. Series B
  (Methodological)}, 34(2):187--220.

\bibitem[Crabbé et~al., 2022]{benchmarkHetTrtEffInterpretability2022}
Crabbé, J., Curth, A., Bica, I., and van~der Schaar, M. (2022).
\newblock Benchmarking heterogeneous treatment effect models through the lens
  of interpretability.

\bibitem[Cuellar, 2022]{Cuellar2022}
Cuellar, M. (2022).
\newblock {\em Causes of Effects and Effects of Causes}, pages 211--233.
\newblock Springer International Publishing, Cham.

\bibitem[Cuellar and Kennedy, 2018]{Cuellar2018-tk}
Cuellar, M. and Kennedy, E. (2018).
\newblock A nonparametric projection-based estimator for the probability of
  causation, with application to water sanitation in kenya.
\newblock {\em SSRN Electron. J.}

\bibitem[Dasgupta et~al., 2006]{Dasgupta2006-Algorithms}
Dasgupta, S., Papadimitriou, C.~H., and Vazirani, U.~V. (2006).
\newblock {\em Algorithms}.
\newblock McGraw-Hill Professional, New York, NY.

\bibitem[Dawid and Musio, 2022]{dawidMusio2022}
Dawid, A.~P. and Musio, M. (2022).
\newblock Effects of causes and causes of effects.
\newblock {\em Annual Review of Statistics and Its Application}, 9(1):261--287.

\bibitem[{Diemert Eustache, Betlei Artem} et~al.,
  2018]{CriteoUpliftDiemert2018}
{Diemert Eustache, Betlei Artem}, Renaudin, C., and Massih-Reza, A. (2018).
\newblock A large scale benchmark for uplift modeling.
\newblock In {\em Proceedings of the AdKDD and TargetAd Workshop, KDD,
  London,United Kingdom, August, 20, 2018}. ACM.

\bibitem[Dvoretzky et~al., 1956]{DKWIneq1956}
Dvoretzky, A., Kiefer, J., and Wolfowitz, J. (1956).
\newblock {Asymptotic Minimax Character of the Sample Distribution Function and
  of the Classical Multinomial Estimator}.
\newblock {\em The Annals of Mathematical Statistics}, 27(3):642 -- 669.

\bibitem[D’Amour et~al., 2021]{overlap2021}
D’Amour, A., Ding, P., Feller, A., Lei, L., and Sekhon, J. (2021).
\newblock {Overlap in observational studies with high-dimensional covariates}.
\newblock {\em Journal of Econometrics}, 221(2):644--654.

\bibitem[Efron and Tibshirani, 1994]{Efron1994-oc}
Efron, B. and Tibshirani, R.~J. (1994).
\newblock {\em An introduction to the bootstrap}.
\newblock Chapman and Hall/CRC.

\bibitem[Fan and Park, 2010]{fanSooPark2010}
Fan, Y. and Park, S.~S. (2010).
\newblock Sharp bounds on the distribution of treatment effects and their
  statistical inference.
\newblock {\em Econometric Theory}, 26(3):931--951.

\bibitem[Fan and Park, 2012]{fanSooPark2012}
Fan, Y. and Park, S.~S. (2012).
\newblock Confidence intervals for the quantile of treatment effects in
  randomized experiments.
\newblock {\em Journal of Econometrics}, 167(2):330--344.
\newblock Fourth Symposium on Econometric Theory and Applications (SETA).

\bibitem[Fay et~al., 2018]{fayMannU2018}
Fay, M.~P., Brittain, E.~H., Shih, J.~H., Follmann, D.~A., and Gabriel, E.~E.
  (2018).
\newblock Causal estimands and confidence intervals associated with
  wilcoxon-mann-whitney tests in randomized experiments.
\newblock {\em Statistics in Medicine}, 37(20):2923--2937.

\bibitem[Fernández-Loría and Provost, 2022]{causalClassification2022}
Fernández-Loría, C. and Provost, F. (2022).
\newblock Causal classification: Treatment effect estimation vs. outcome
  prediction.
\newblock {\em Journal of Machine Learning Research}, 23(59):1--35.

\bibitem[Firpo and Ridder, 2019]{FIRPO2019210}
Firpo, S. and Ridder, G. (2019).
\newblock Partial identification of the treatment effect distribution and its
  functionals.
\newblock {\em Journal of Econometrics}, 213(1):210--234.
\newblock Annals: In Honor of Roger Koenker.

\bibitem[Firpo and Ridder, 2010]{firpo2010}
Firpo, S.~P. and Ridder, G. (2010).
\newblock {Bounds on functionals of the distribution treatment effects}.
\newblock Textos para discussão 201, FGV EESP - Escola de Economia de São
  Paulo, Fundação Getulio Vargas (Brazil).

\bibitem[Fisher, 1935]{Fisher1935}
Fisher, R.~A. (1935).
\newblock {\em The design of experiments}.
\newblock Oliver and Boyd, Edinburgh.

\bibitem[Frandsen and Lefgren, 2021]{frandsenPositiveCor2021}
Frandsen, B.~R. and Lefgren, L.~J. (2021).
\newblock Partial identification of the distribution of treatment effects with
  an application to the knowledge is power program (kipp).
\newblock {\em Quantitative Economics}, 12(1):143--171.

\bibitem[Frank et~al., 1987]{Frank1987BestpossibleBF}
Frank, M.~J., Nelsen, R.~B., and Schweizer, B. (1987).
\newblock Best-possible bounds for the distribution of a sum — a problem of
  kolmogorov.
\newblock {\em Probability Theory and Related Fields}, 74:199--211.

\bibitem[Fr\'echet, 1935]{FrechetMaurice1935Gdtd}
Fr\'echet, M. (1935).
\newblock G\'en\'eralisation du th\'eor\`eme des probabilit\'es totales.
\newblock {\em Fundamenta mathematicae}, 25:379--387.

\bibitem[Fr\'echet, 1960]{FrechetM.1960Sltd}
Fr\'echet, M. (1960).
\newblock Sur les tableaux dont les marges et des bornes sont données.
\newblock {\em Revue de l'Institut international de statistique},
  28(1/2):10--32.

\bibitem[Funk et~al., 2011]{doublyRobust2011}
Funk, M.~J., Westreich, D., Wiesen, C., Stürmer, T., Brookhart, M.~A., and
  Davidian, M. (2011).
\newblock {Doubly Robust Estimation of Causal Effects}.
\newblock {\em American Journal of Epidemiology}, 173(7):761--767.

\bibitem[Gradu et~al., 2022]{validInferenceAfterDiscovery}
Gradu, P., Zrnic, T., Wang, Y., and Jordan, M.~I. (2022).
\newblock Valid inference after causal discovery.

\bibitem[Greenland et~al., 2020]{Greenland2020-tu}
Greenland, S., Fay, M.~P., Brittain, E.~H., Shih, J.~H., Follmann, D.~A.,
  Gabriel, E.~E., and Robins, J.~M. (2020).
\newblock On causal inferences for personalized medicine: How hidden causal
  assumptions led to erroneous causal claims about the d-value.
\newblock {\em Am. Stat.}, 74(3):243--248.

\bibitem[Gupta et~al., 2020]{distrFreeBinaryReg2020}
Gupta, C., Podkopaev, A., and Ramdas, A. (2020).
\newblock Distribution-free binary classification: prediction sets, confidence
  intervals and calibration.
\newblock In Larochelle, H., Ranzato, M., Hadsell, R., Balcan, M., and Lin, H.,
  editors, {\em Advances in Neural Information Processing Systems}, volume~33,
  pages 3711--3723. Curran Associates, Inc.

\bibitem[Gupta and Ramdas, 2021]{pmlr-v139-gupta21b}
Gupta, C. and Ramdas, A. (2021).
\newblock {Distribution-Free Calibration Guarantees for Histogram Binning
  without Sample Splitting}.
\newblock In Meila, M. and Zhang, T., editors, {\em {Proceedings of the 38th
  International Conference on Machine Learning}}, volume 139 of {\em
  Proceedings of Machine Learning Research}, pages 3942--3952. PMLR.

\bibitem[Hailperin, 1986]{hailperinBoole1986}
Hailperin, T. (1986).
\newblock {\em Boole's logic and probability : a critical exposition from the
  standpoint of contemporary algebra, logic, and probability theory / Theodore
  Hailperin.}
\newblock Studies in logic and the foundations of mathematics ; v. 85.
  North-Holland Pub. Co., Amsterdam, Netherlands ;, 2nd ed., rev. and enl.
  edition.

\bibitem[Hand, 1992]{hand1992comparing}
Hand, D.~J. (1992).
\newblock On comparing two treatments.
\newblock {\em The American Statistician}, 46(3):190--192.

\bibitem[Hern\'an and Robins, 2020]{hernan2020}
Hern\'an, M.~A. and Robins, J. (2020).
\newblock {\em Causal Inference: What If}.
\newblock Boca Raton: Chapman \& Hall/CRC.

\bibitem[Imbens and Rubin, 2015]{imbens_rubin_2015}
Imbens, G.~W. and Rubin, D.~B. (2015).
\newblock {\em Causal Inference for Statistics, Social, and Biomedical
  Sciences: An Introduction}.
\newblock Cambridge University Press.

\bibitem[Jin et~al., 2021]{JinRenCandes2021SensitivityITE}
Jin, Y., Ren, Z., and Candès, E.~J. (2021).
\newblock Sensitivity analysis of individual treatment effects: A robust
  conformal inference approach.

\bibitem[Kallus, 2022]{kallus2022_propHarmed}
Kallus, N. (2022).
\newblock What's the harm? sharp bounds on the fraction negatively affected by
  treatment.

\bibitem[Kennedy, 2020]{kennedy2022OptimalDoublyRobustCATE}
Kennedy, E.~H. (2020).
\newblock Towards optimal doubly robust estimation of heterogeneous causal
  effects.

\bibitem[Kneib et~al., 2021]{KNEIB2021DistrRegr}
Kneib, T., Silbersdorff, A., and Säfken, B. (2021).
\newblock Rage against the mean – a review of distributional regression
  approaches.
\newblock {\em Econometrics and Statistics}.

\bibitem[Koenker and Bassett, 1978]{regressionQuantilesKoenkerBasset1978}
Koenker, R. and Bassett, G. (1978).
\newblock Regression quantiles.
\newblock {\em Econometrica}, 46(1):33--50.

\bibitem[Koenker et~al., 2013]{KoenkerLeoratoDistrRegr2013}
Koenker, R., Leorato, S., and Peracchi, F. (2013).
\newblock Distributional vs. quantile regression.
\newblock {\em SSRN Electronic Journal}.

\bibitem[Kumar et~al., 2019]{verifiedUncerCalib2019}
Kumar, A., Liang, P.~S., and Ma, T. (2019).
\newblock Verified uncertainty calibration.
\newblock In Wallach, H., Larochelle, H., Beygelzimer, A., d\textquotesingle
  Alch\'{e}-Buc, F., Fox, E., and Garnett, R., editors, {\em Advances in Neural
  Information Processing Systems}, volume~32. Curran Associates, Inc.

\bibitem[K{\"u}nzel et~al., 2019]{Kunzel4156}
K{\"u}nzel, S.~R., Sekhon, J.~S., Bickel, P.~J., and Yu, B. (2019).
\newblock Metalearners for estimating heterogeneous treatment effects using
  machine learning.
\newblock {\em Proceedings of the National Academy of Sciences},
  116(10):4156--4165.

\bibitem[Lehmann, 1966]{Lehman1966}
Lehmann, E.~L. (1966).
\newblock Some concepts of dependence.
\newblock {\em The Annals of Mathematical Statistics}, 37(5):1137--1153.

\bibitem[Lei and Candès, 2021]{conformalITE2021}
Lei, L. and Candès, E.~J. (2021).
\newblock {Conformal inference of counterfactuals and individual treatment
  effects}.
\newblock {\em Journal of the Royal Statistical Society Series B},
  83(5):911--938.

\bibitem[Li et~al., 2021]{LiChenTong2021_NPFeatureRanking}
Li, J.~J., Chen, Y.~E., and Tong, X. (2021).
\newblock A flexible model-free prediction-based framework for feature ranking.
\newblock {\em Journal of Machine Learning Research}, 22(124):1--54.

\bibitem[Lu et~al., 2015]{Lu2015OrdinalIndvidTreatmentEO}
Lu, J., Ding, P., and Dasgupta, T. (2015).
\newblock Treatment effects on ordinal outcomes: Causal estimands and sharp
  bounds.
\newblock {\em Journal of Educational and Behavioral Statistics}, 43:540 --
  567.

\bibitem[Makarov, 1982]{makarovBounds}
Makarov, G.~D. (1982).
\newblock Estimates for the distribution function of a sum of two random
  variables when the marginal distributions are fixed.
\newblock {\em Theory of Probability \& Its Applications}, 26(4):803--806.

\bibitem[Massart, 1990]{massartDKWIneq1990}
Massart, P. (1990).
\newblock {The Tight Constant in the Dvoretzky-Kiefer-Wolfowitz Inequality}.
\newblock {\em The Annals of Probability}, 18(3):1269 -- 1283.

\bibitem[McCullagh and Nelder, 2019]{McCullagh2019-qz}
McCullagh, P. and Nelder, J.~A. (2019).
\newblock {\em Generalized Linear Models}.
\newblock Routledge.

\bibitem[Mueller et~al., 2022]{muellerLiPearl2021PNSCovariates}
Mueller, S., Li, A., and Pearl, J. (2022).
\newblock Causes of effects: Learning individual responses from population
  data.
\newblock In Raedt, L.~D., editor, {\em Proceedings of the Thirty-First
  International Joint Conference on Artificial Intelligence, {IJCAI-22}}, pages
  2712--2718. International Joint Conferences on Artificial Intelligence
  Organization.
\newblock Main Track.

\bibitem[Naaman, 2021]{NaamanDKWIneq2021}
Naaman, M. (2021).
\newblock On the tight constant in the multivariate
  dvoretzky–kiefer–wolfowitz inequality.
\newblock {\em Statistics \& Probability Letters}, 173:109088.

\bibitem[Nie and Wager, 2020]{nieWager}
Nie, X. and Wager, S. (2020).
\newblock {Quasi-oracle estimation of heterogeneous treatment effects}.
\newblock {\em Biometrika}, 108(2):299--319.

\bibitem[Padilla et~al., 2021]{padilla2021causal}
Padilla, O. H.~M., Ding, P., Chen, Y., and Ruiz, G. (2021).
\newblock A causal fused lasso for interpretable heterogeneous treatment
  effects estimation.

\bibitem[Parikh et~al., 2022]{validatingCauslInfMethods22}
Parikh, H., Varjao, C., Xu, L., and Tchetgen, E.~T. (2022).
\newblock Validating causal inference methods.
\newblock In Chaudhuri, K., Jegelka, S., Song, L., Szepesvari, C., Niu, G., and
  Sabato, S., editors, {\em Proceedings of the 39th International Conference on
  Machine Learning}, volume 162 of {\em Proceedings of Machine Learning
  Research}, pages 17346--17358. PMLR.

\bibitem[Pearl, 1999]{Pearl1999-sh}
Pearl, J. (1999).
\newblock Probabilities of causation: Three counterfactual interpretations and
  their identification.
\newblock {\em Synthese}, 121(1):93--149.

\bibitem[Pearl, 2009]{pearlCausality}
Pearl, J. (2009).
\newblock {\em Causality: Models, Reasoning and Inference}.
\newblock Cambridge University Press, USA, 2nd edition.

\bibitem[Robins et~al., 1994]{robins1994estimation}
Robins, J.~M., Rotnitzky, A., and Zhao, L.~P. (1994).
\newblock Estimation of regression coefficients when some regressors are not
  always observed.
\newblock {\em Journal of the American statistical Association},
  89(427):846--866.

\bibitem[Romano et~al., 2019]{romanoConformal2019}
Romano, Y., Patterson, E., and Cand\`es, E.~J. (2019).
\newblock Conformalized quantile regression.
\newblock {\em Advances in Neural Information Processing Systems 32}, pages
  3543--3553.

\bibitem[Rubin, 1974]{Rubin1974EstimatingCE}
Rubin, D. (1974).
\newblock Estimating causal effects of treatments in randomized and
  nonrandomized studies.
\newblock {\em Journal of Educational Psychology}, 66:688--701.

\bibitem[R{\"u}schendorf, 1981]{Rschendorf1981SharpnessOfFrechetBoole}
R{\"u}schendorf, L. (1981).
\newblock Sharpness of fr{\'e}chet-bounds.
\newblock {\em Zeitschrift f{\"u}r Wahrscheinlichkeitstheorie und Verwandte
  Gebiete}, 57:293--302.

\bibitem[Schober and Vetter, 2018]{Schober2018_TTE}
Schober, P. and Vetter, T.~R. (2018).
\newblock Survival analysis and interpretation of time-to-event data: The
  tortoise and the hare.
\newblock {\em Anesth. Analg.}, 127(3):792--798.

\bibitem[Shafer and Vovk, 2008]{conformalShaferVovk08a}
Shafer, G. and Vovk, V. (2008).
\newblock A tutorial on conformal prediction.
\newblock {\em Journal of Machine Learning Research}, 9(12):371--421.

\bibitem[Splawa-Neyman et~al., 1990]{splawa-neyman1990}
Splawa-Neyman, J., Dabrowska, D.~M., and Speed, T.~P. (1990).
\newblock On the application of probability theory to agricultural experiments.
  essay on principles. section 9.
\newblock {\em Statist. Sci.}, 5(4):465--472.

\bibitem[Stitelman and van~der Laan, 2010]{Stitelman2010_TTE_TMLE}
Stitelman, O.~M. and van~der Laan, M.~J. (2010).
\newblock Collaborative targeted maximum likelihood for time to event data.
\newblock {\em Int. J. Biostat.}, 6(1):Article 21.

\bibitem[Sur and Candès, 2019]{modernHighDLogisRegr_SurCandes2019}
Sur, P. and Candès, E.~J. (2019).
\newblock A modern maximum-likelihood theory for high-dimensional logistic
  regression.
\newblock {\em Proceedings of the National Academy of Sciences},
  116(29):14516--14525.

\bibitem[Sverdrup et~al., 2020]{poliTreeSofware}
Sverdrup, E., Kanodia, A., Zhou, Z., Athey, S., and Wager, S. (2020).
\newblock {policytree: Policy learning via doubly robust empirical welfare
  maximization over trees}.
\newblock {\em Journal of Open Source Software}, 5(50):2232.

\bibitem[Tian and Pearl, 2000]{Tian2000-ul}
Tian, J. and Pearl, J. (2000).
\newblock Probabilities of causation: Bounds and identification.
\newblock {\em Ann. Math. Artif. Intell.}, 28(1/4):287--313.

\bibitem[Tibshirani et~al., 2019]{conformal2019DistrShift}
Tibshirani, R.~J., Barber, R.~F., Cand\`{e}s, E.~J., and Ramdas, A. (2019).
\newblock {\em Conformal Prediction under Covariate Shift}.
\newblock Curran Associates Inc., Red Hook, NY, USA.

\bibitem[Tong et~al., 2018]{TongFengLi2018_NPClassification}
Tong, X., Feng, Y., and Li, J.~J. (2018).
\newblock Neyman-pearson classification algorithms and np receiver operating
  characteristics.
\newblock {\em Science Advances}, 4(2):eaao1659.

\bibitem[van~der Vaart and Wellner, 1996]{vanderVaart1996}
van~der Vaart, A.~W. and Wellner, J.~A. (1996).
\newblock {\em Weak Convergence and Empirical Processes: With Applications to
  Statistics}.
\newblock Springer New York, New York, NY.

\bibitem[Vovk et~al., 2005]{VovkGammermanShafer2005Conformal}
Vovk, V., Gammerman, A., and Shafer, G. (2005).
\newblock {\em Algorithmic Learning in a Random World}.

\bibitem[Wager and Athey, 2018]{causalForests}
Wager, S. and Athey, S. (2018).
\newblock Estimation and inference of heterogeneous treatment effects using
  random forests.
\newblock {\em Journal of the American Statistical Association},
  113(523):1228--1242.

\bibitem[Wainwright, 2019]{wainwright_2019}
Wainwright, M.~J. (2019).
\newblock {\em High-Dimensional Statistics: A Non-Asymptotic Viewpoint}.
\newblock Cambridge Series in Statistical and Probabilistic Mathematics.
  Cambridge University Press.

\bibitem[Williamson and Downs, 1990]{WILLIAMSON199089}
Williamson, R.~C. and Downs, T. (1990).
\newblock Probabilistic arithmetic. i. numerical methods for calculating
  convolutions and dependency bounds.
\newblock {\em International Journal of Approximate Reasoning}, 4(2):89--158.

\bibitem[Xu and Yadlowsky, 2022]{calibrationHeteroTEs_XuYadlowsky2022}
Xu, Y. and Yadlowsky, S. (2022).
\newblock Calibration error for heterogeneous treatment effects.
\newblock In Camps{-}Valls, G., Ruiz, F. J.~R., and Valera, I., editors, {\em
  International Conference on Artificial Intelligence and Statistics, {AISTATS}
  2022, 28-30 March 2022, Virtual Event}, volume 151 of {\em Proceedings of
  Machine Learning Research}, pages 9280--9303. {PMLR}.

\bibitem[Yadlowsky et~al., 2021]{YadlowskyWagerTrtmntPrioritizationRules2021}
Yadlowsky, S., Fleming, S., Shah, N., Brunskill, E., and Wager, S. (2021).
\newblock Evaluating treatment prioritization rules via rank-weighted average
  treatment effects.

\bibitem[Ye et~al., 2021]{yeChenPadilla2021Scores}
Ye, S.~S., Chen, Y., and Padilla, O. H.~M. (2021).
\newblock 2d score based estimation of heterogeneous treatment effects.

\bibitem[Yin et~al., 2021]{YinShiWangBlei2021SensitivityITE}
Yin, M., Shi, C., Wang, Y., and Blei, D.~M. (2021).
\newblock Conformal sensitivity analysis for individual treatment effects.

\bibitem[Zadrozny and Elkan, 2002]{unifMassBinning2001}
Zadrozny, B. and Elkan, C. (2002).
\newblock {Transforming Classifier Scores into Accurate Multiclass Probability
  Estimates}.
\newblock In {\em Proceedings of the Eighth ACM SIGKDD International Conference
  on Knowledge Discovery and Data Mining}, KDD '02, page 694–699, New York,
  NY, USA. Association for Computing Machinery.

\end{thebibliography}
\bibliographystyle{apalike}


\newpage
\appendix
\onecolumn


\section{The Makarov bounds\label{append:makBoundsDiscussion}}
\subsection{Formal Statement}
We present the bounds as stated in \citet{WILLIAMSON199089}'s {Theorem 2}.
\begin{lemma}[The Makarov Bounds \label{lem:makBounds}]

Uniformly across all possible, unknown joint distributions 
 \[
 (V_1,V_2)\sim \pr{V_1\leq v_1,V_2\leq v_2}
 \]
 having fixed marginal $\cdf$s, $F_1(v_1)=\pr{V_1\leq v_1}$ and $F_2(v_2)=\pr{V_2\leq v_1}$, the $\cdf$ of $V_1-V_2$ evaluated at $\delta\in\mathbb{R}$ satisfies:
\begin{equation}
\begin{aligned}
F^L(\delta)&\leq\ && \pr{V_1-V_2\leq \delta}&\leq\ &F^U(\delta),\\
\end{aligned}
\end{equation}
where 
\[\begin{aligned}
F^L(\delta)&=\ & -\max\left[\sup_{a,b:\ a+b=\delta } \{F_1(a)-F_2(-b)\},0\right]
\end{aligned}\]
and
\[\begin{aligned}
F^U(\delta)&=\ &1+\min\left[ \inf_{a,b:\ a+b=\delta } \{F_1(a)-F_2(-b) \},0\right].
\end{aligned}\]

\end{lemma}

Lemma \ref{lem:makBounds} was first proved in \citet{makarovBounds} to bound the distribution of a sum of two random variables. We present this result for subtraction, which is a simple extension, as $V_1-V_2$ is technically the sum of two random variables $V_1$ and $(-V_2)$. Lemma \ref{lem:makBounds}'s proof was later extended in \cite{Frank1987BestpossibleBF} and \citet{WILLIAMSON199089}, who also seek bounds on the distribution of other binary operations on $V_1$ and $V_2$, like their difference, product, and their ratio, under minimal distributional assumptions.

\subsection{Equivalent forms of the bounds}

The terms $F^L(\delta)$ and $F^U(\delta)$ in Lemma \ref{lem:makBounds} can be rewritten. 

Consider a one-to-one change of variables $(a,b)\mapsto(u+\delta/2,-u+\delta/2)$. With it, we have equivalently:
\[
F^L(\delta)\ =\ -\max\left[\sup_{u} \{F_1(u+\delta/2)-F_2(u-\delta/2)\},0\right]
\]
along with
\[
F^U(\delta)\ =\ 1+\min\left[ \inf_{u} \{F_1(u+\delta/2)-F_2(u-\delta/2)\},0\right].
\]
This is in line with what we have written in Section \ref{sec:margITEMakBnds} and \ref{sec:condITEMakBnds} of the main text.

Consider instead the one-to-one change of variables $(a,b)\mapsto(u,-u+\delta)$. With it, we have equivalently:
\[
F^L(\delta)\ =\ -\max\left[\sup_{u} \{F_1(u)-F_2(u-\delta)\},0\right]
\]
along with
\[
F^U(\delta)\ =\ 1+\min\left[ \inf_{u} \{F_1(u)-F_2(u-\delta)\},0\right].
\]
This is in line with \citet{fanSooPark2010}'s {Lemma 2.1} and {Equation (2) and (3)}, and it also agrees with the alternative form given in {Equations (21) and (22)} of \citet{WILLIAMSON199089}. 

Moreover, it is straightforward to see that:
\[
1-F^U(\delta)\ \leq\ \pr{U_1-U_0>\delta}\ \leq\ 1-F^L(\delta).
\]
We make use of this in the main text when bounding $\pibt$. 

\section{Appendix to Section \ref{sec:margITEMakBnds}}

\subsection{Computational Complexity\label{append:compComplex}}

The reasoning for the computational complexity claim in Remark \ref{rem:compComplex} is as follows. Consider the mapping 
\[
\hat{G}_{n}(y,\delta,x):=\hat{F}_{1n}(y+\delta/2)\mid x-\hat{F}_{0n}(y-\delta/2\mid x).
\]
The bound estimators in \eqref{eqn:lwrEstRE} and \eqref{eqn:uprEstRE} are each a functional of $\hat{G}_n(\cdot,\delta)$ for fixed $(\delta,x)$. The knots of $\hat{G}(\cdot,\delta,x)$ are given by the set 
\[
\mathcal{Y}_{\delta,x, n}:=\{Y_i(w)-\delta/2:\ W_i=0,\ell(X_i)=\ell(x)\}\cup\{Y_i(w)+\delta/2:\ W_i=1,\ell(X_i)=\ell(x)\}.
\]
This means that the infimum and supremum used to specify the bound estimators in \eqref{eqn:lwrEstRE} and \eqref{eqn:uprEstRE} are really a minimum and maximum, respectively, taken by indexing the finite set $\mathcal{Y}_{\delta,x, n}$.

Next, consider that sorting the list $\curl{Y_{i}:\ W_i=w,\ellxi=\ellx}$ has worst time computational complexity $O(n_w(x)\log(n_w(x)))$ with a procedure such as merge sort \citep{Dasgupta2006-Algorithms}. Given the lists $\curl{Y_{i}:\ W_i=0,\ellxi=\ellx}$ and $\curl{Y_{i}:\ W_i=1,\ellxi=\ellx}$ in sorted form, the evaluation of $\hat{G}(y,\delta,x)$ for one $y\in\mathcal{Y}_{\delta,x, n}$ is no worse than $O\curl{\log(n_0(x))+\log(n_1(x) )}$ with a procedure making use of binary search \citep{Dasgupta2006-Algorithms}. Evaluating $\hat{G}(y,x,\delta)$ at each point $y\in\mathcal{Y}_{\delta,x, n}$ then brings us to the stated overall computational complexity of $O[ \parenth{n_0(x)+n_1(x)}\curl{\log(n_0(x))+\log(n_1(x))} ]$. Finally, the minimum and maximum calculation on the list of evaluations $\{\hat{G}(y,x,\delta):\ y\in\mathcal{Y}_{\delta,x, n}\}$ add $O(n_0(x)+n_1(x))$ complexity, so we retain the overall complexity stated in Remark \ref{rem:compComplex}. 

\subsection{A simulation study\label{sec:margCaseSimulations}}

\begin{figure}[htbp]
  \centering
  \includegraphics[width=0.45\textwidth]{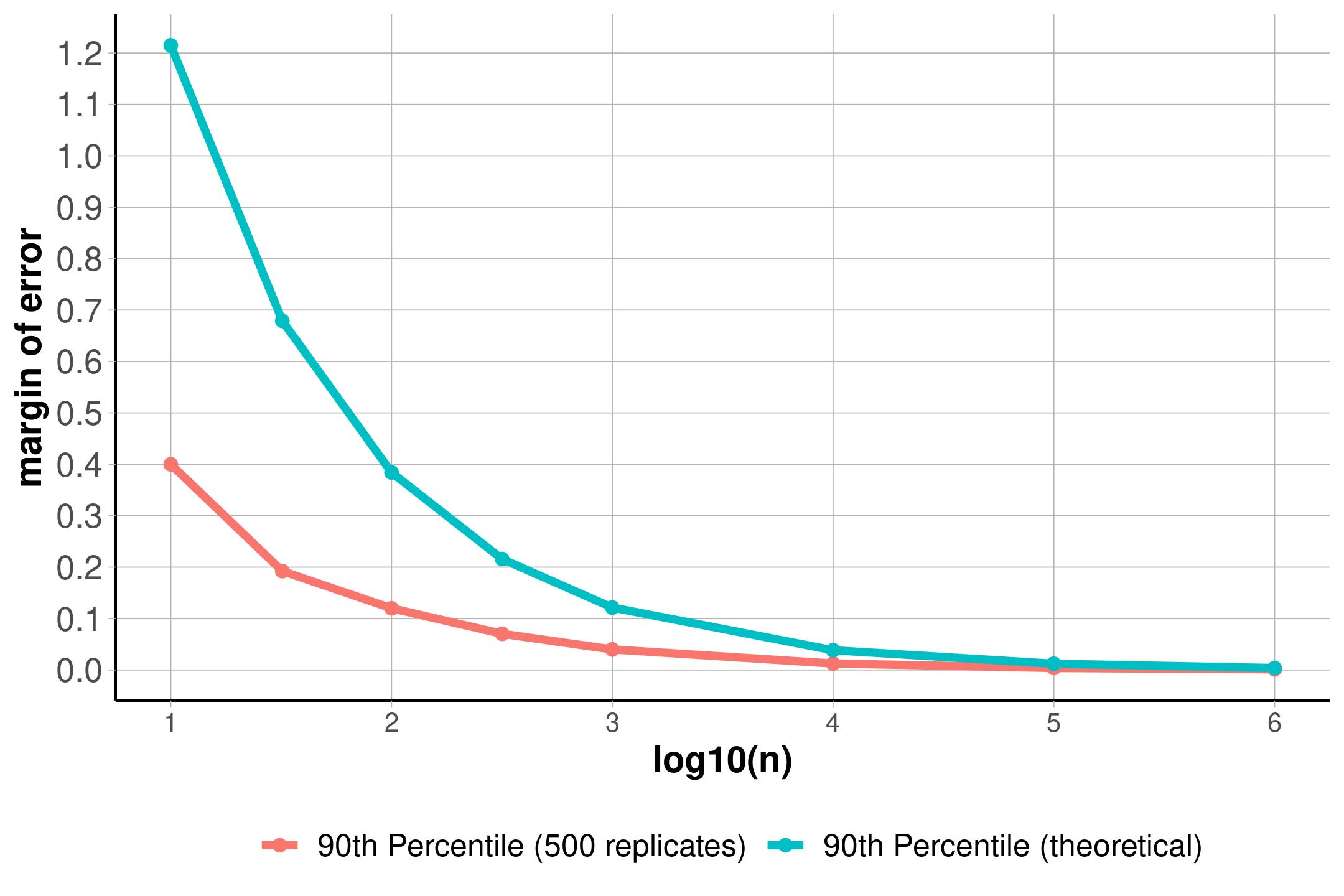}
    \caption{Across 500 replicates at each $n$, we compare the 90th percentile of the error defined in Equation \eqref{eqn:errorSims} against the margin of error given by Proposition \ref{prop:mainResultMargPIBT} at the 90\% confidence level. We hold the signal to noise ratio and propensity score fixed.}
    \label{fig:simsMargCase_MoreN_MoreDelta}
\end{figure}

\begin{figure}
    \centering
    \includegraphics[width=0.45\textwidth]{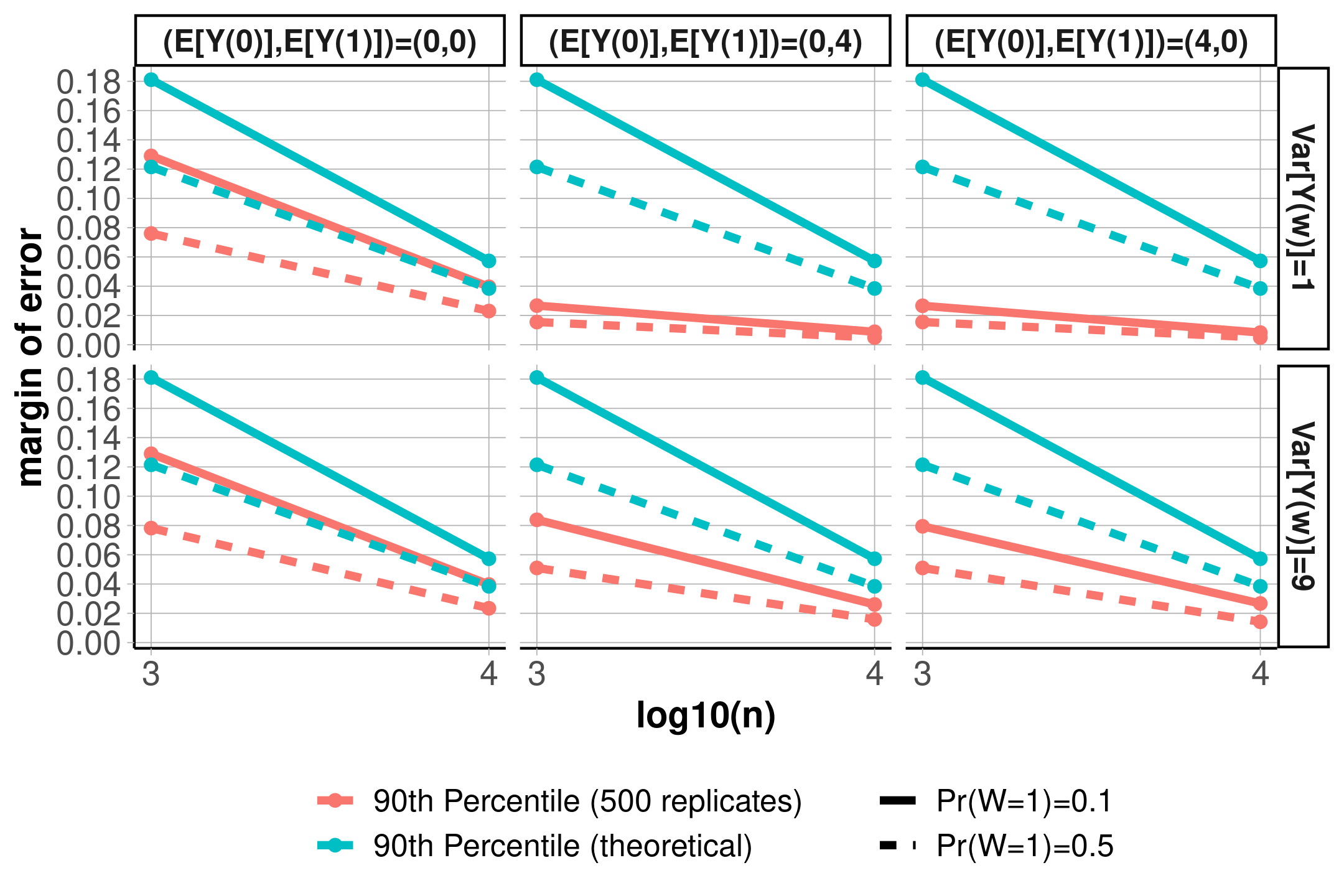}
    \caption{Across 500 replicates at each $n$ and $\delta$, we compare the 90th percentile of the error defined in Equation \eqref{eqn:errorSims2} against the margin of error given by Proposition \ref{prop:mainResultMargPIBT} at the 90\% confidence level. We fix $\delta=0$, vary sample size, the signal to noise ratio, and the propensity score.}
    \label{fig:DiffSig2Noise}
\end{figure}

We now conduct some simulations to demonstrate that Proposition \ref{prop:mainResultMargPIBT} provides the appropriate coverage. First, we estimate bounds on $\pibt$ at multiple thresholds $\delta$. Then we fix $\delta$ and estimate $\pibt$ based on synthetic data generated according to differing multiple signal to noise ratios and propensity scores. 


In the first setting, our data is generated according to \[
Y_i(w)\sim\mathcal{N}(0,1)\ (w=0,1)
\]
and $\pr{W_i=1}=0.5$. Moreover, our bound estimators are based on sample size and threshold given by their respective choices: $n= 10, \lceil{10^{1.5}\rceil},10^2,\lceil{10^{2.5}\rceil}, 10^3, 10^4, 10^5, 10^6$ and $\delta=0,\dots,5$. For each $n$, we keep track of the estimation error
\begin{equation}\label{eqn:errorSims}
\max_{\delta=0,\dots,5}\abs{\hat{\theta}^L(\delta)-{\theta}^L(\delta) }\vee\abs{\hat{\theta}^U(\delta)-{\theta}^U(\delta) }
\end{equation}
across $500$ replicates. Figure \ref{fig:simsMargCase_MoreN_MoreDelta} summarizes these results by comparing the 90th percentile of this error across the 500 replicates to the margin of error given by Proposition \ref{prop:mainResultMargPIBT} at the 90\% confidence level. As suggested by this figure, our joint coverage of $\theta^L(\delta)$ and $\theta^U(\delta)$ for each $\delta=0,1,2,3,4,5$ using Proposition \ref{prop:mainResultMargPIBT} appears to be conservative. In fact, across the 500 replicates for each $n$ we tried, the actual joint coverage is no less than 99.8\%.

{}

For the nominal coverage guarantee for the bound estimators using Proposition \ref{prop:mainResultMargPIBT}, we now seek to understand the influence of varying values of the average treatment effect, $\mu_1-\mu_0$, differing amounts of potential outcomes' variation, $\sigma^2$, and the propensity score, $\pr{W_i=1}$. We fix $\delta=0$ in our estimation of the bounds on $\pibt$. For this second synthetic data scenario, we generate the potential outcomes marginally according to:
\[
Y_i(w)\sim\mathcal{N}(\mu_w,\sigma^2).
\]
We consider $500$ replicates for every combination of simulation parameters denoted in Fig. \ref{fig:DiffSig2Noise} and keep track of the error
\begin{equation}\label{eqn:errorSims2}
\abs{\hat{\theta}^L(0)-{\theta}^L(0) }\vee\abs{\hat{\theta}^U(0)-{\theta}^U(0) }.
\end{equation}

As with the previous synthetic example, it seems that the theoretical margin of error for a 90\% confidence level given by Proposition \ref{prop:mainResultMargPIBT} is conservative compared to the 90th percentile of the bound estimates' margin of error. Again, the actual joint coverage of the bound estimators is at or very close to 100\%.

Recall that the margin of error in Proposition \ref{prop:mainResultMargPIBT} only depends on sample size and target confidence level. Correspondingly, the only change in the theoretical margin of error is when the propensity score changes between $0.1$ and $0.5$, corresponding to the bold-blue and dashed-blue lines in Fig. \ref{fig:DiffSig2Noise}, respectively. Observe now the red lines in Fig. \ref{fig:DiffSig2Noise}, corresponding to the 90th percentile of the estimation error across the 500 replicates. As expected, the propensity score also affects estimation error, with the balanced case of $\pr{W_i=1}=0.5$ (dashed-red lines) generally giving better results compared to the imbalanced case of $\pr{W_i=1}=0.1$ (bold-red lines). Denote signal to noise as 
\[
\abs{\mathbb{E}\curl{Y_i(1)-Y_i(0)}}/\sqbrack{{\rm var}\curl{Y_i(0)}+{\rm var}\curl{Y_i(1)}}^{1/2}.
\]
We see from the red lines in Fig. \ref{fig:DiffSig2Noise} a smaller estimation error when the signal to noise is $2^{3/2}$ compared to when it is $0$ or $2^{ 3/2} /3$.

\section{$\pibt$ bounds with pre-treatment covariates\label{sec:condITEMakBnds}}
\subsection{Overview}

We now seek to estimate bounds on the unidentifiable probability an individual benefits from treatment ($\pibt$) in pre-treatment stratum $X_i=x$:
\[
\theta(\delta,x)={\rm pr}\left\{Y_i(1)-Y_i(0)>\delta\mid X_i=x\right\}.
\]
Could it be known, this quantity is helpful to understand whether the benefit of receiving treatment varies across pre-treatment covariate strata. Denote the target bounds as $\theta^{L}(\delta,x)$ and $\theta^{U}(\delta,x)$, the lower and upper bound, respectively. They satisfy:
\[
\theta^L(\delta,x)\leq \theta(\delta,x)\leq\theta^{U}(\delta,x).
\]
And denote the corresponding estimators as $\hat{\theta}^{L}(\delta,x)$ and $\hat{\theta}^{U}(\delta,x)$, respectively. We would like a guarantee about how close $\{\hat{\theta}^L(\delta,x),\hat{\theta}^U(\delta,x)\}$ is to $\curl{\theta^L(\delta,x),\theta^U(\delta,x)}$. 
\begin{remark}[Large enough sample at a covariate stratum?]

Importantly, should a large enough sample be collected at stratum $x$ of the pre-treatment covariates, Proposition \ref{prop:mainResultMargPIBT} can be applied for a frequentist confidence statement about $\theta(\delta,x)$ using the analogous interpretation in Remark \ref{rem:practicalInterpretation}. The rest of this section is useful for the case that a large enough sample is not collected for some or all of the pre-treatment covariate strata of interest. 
\end{remark}

Proposition \ref{prop:mainResConditionalGeneral} is the main non-asymptotic, model-free result for this setting. Proposition \ref{prop:mainResRegrResids} is the adaptation of Proposition \ref{prop:mainResConditionalGeneral} to a case where we strategically use regression residuals to estimate the bounds on $\pibt$. For this model-free approach with regression residuals, we show how a confidence statement about the conditional $\pibt$ bound estimators can be written in terms of a target confidence level that is adjusted according to how accurate the arbitrary regression function estimator is. In Corollary \ref{cor:effCondBounds}, we demonstrate how the statement written in this manner implies that the conditional bound estimators are as statistically efficient as the regression function estimator of choice. Moreover, Proposition \ref{prop:CIsGaussResids} adopts the more general Proposition \ref{prop:mainResRegrResids} to the canonical linear regression case. 

\subsection{Existing work: $\pibt$ conditional on pre-treatment covariates}

\citet{fanSooPark2010} also discuss the conditional bounds for the conditional individual treatment effect distribution at the population-level along with a brief discussion of possible estimation approaches. The appendix of \citet{frandsenPositiveCor2021} also mentions a generalisation of the mutual stochastic increasing-ness assumption in order to arrive at bounds for $1-\theta(\delta,x)$. In the context of ordinal outcomes, using Makarov's bounds to study the individual treatment effect, \citet{Lu2015OrdinalIndvidTreatmentEO} also consider the case we would like to condition on covariates. All three works suggest some form of distributional regression, the semi-parametric estimation of a conditional $\cdf$ \citep{KoenkerLeoratoDistrRegr2013,chernozhukovDiazDistrTrtmntEff2013,KNEIB2021DistrRegr}. We extend their discussion on covariate conditioning with a discussion on theoretical guarantees and how to conduct statistical inference with the bound estimators.

Related to our use of pre-treatment covariates, \citet{conformalITE2021} develops prediction intervals for the individual treatment effect based on quantile regression \citep{regressionQuantilesKoenkerBasset1978} with strategic calibration using conformal inference \citep{VovkGammermanShafer2005Conformal,conformalShaferVovk08a,conformal2019DistrShift}. Moreover, this work is extended to scenarios where unobserved confounding is possible \citep{YinShiWangBlei2021SensitivityITE,JinRenCandes2021SensitivityITE}. Our work here is complementary to these advances, in analogy to the inverse relation between quantiles and the $\cdf$ of a distribution. With respect to theoretical guarantees, Proposition \ref{prop:mainResRegrResids} below is with respect to the supremum deviation across inputted pre-treatment covariate levels, whereas the conformal guarantees of \citet{conformalITE2021} are with respect to any single randomly generated covariate level. This is a subtle but important difference: one may like inference about heterogeneity in individual treatments effects to extend simultaneously to multiple individuals with fixed (non-random) covariate levels, not necessarily a single random individual with a random pre-treatment covariate value. On the other hand, our work does not necessarily extend to a target population beyond that represented by our training sample, and one of the main results here (Proposition \ref{prop:mainResRegrResids}) makes use of a regularity condition on regression residuals that quantile regression generally avoids.  

Related to the use of quantiles to infer individual treatment effects, \citet{fanSooPark2010} and \citet{fanSooPark2012} discuss estimators for the sharp bounds on the quantiles of the individual treatment effect distribution. 




\subsection{The target bounds on conditional $\pibt$, their estimators, and the main result}

The bounds $\theta^L(\delta,x)$ and $\theta^U(\delta,x)$ make use of the conditional $\cdf$s ($w=0,1$):
\[
F_w(y\mid x):={\rm pr}\curl{Y_i(w)\leq y \mid X_i=x}=\prc{Y_i\leq y}{X_i=x,W_i=w},
\]
with the second equality due to Assumption \ref{assump:strongCondIgno} and consistency. Explicitly, due to Lemma \ref{lem:makBounds} in the Appendix \ref{append:makBoundsDiscussion}, we have:
\[
\theta^L(\delta,x) = -\min\left[\inf_y\left\{ F_1(y+\delta/2\mid x)-F_0(y-\delta/2\mid x) \right\},0\right]
\]
along with
\[
\theta^U(\delta,x) = 1-\max\left[\sup_y\left\{ F_1(y+\delta/2\mid x)-F_0(y-\delta/2\mid x) \right\},0\right]
\]
at the population-level. Denote \[
G(y,\delta,x):= F_1(y+\delta/2\mid x)-F_0(y-\delta/2\mid x),
\]
and its corresponding estimator as $\hat{G}_n(y,\delta,x)$ based on the training sample. In practice, one may specify $\hat{G}_n(y,\delta,x)$ as the difference of two conditional $\cdf$ estimators as in Corollary \ref{cor:mainResDistrRegr} below. The plug-in estimators for the lower bound and upper bounds, respectively, will be:
\begin{equation}\label{eqn:plugInCondBoundL}
\hat{\theta}^L(\delta,x) = -\min\left[\inf_y\left\{ \hat{G}(y,\delta,x) \right\},0\right]
\end{equation}
along with
\begin{equation}\label{eqn:plugInCondBoundU}
\hat{\theta}^U(\delta,x) = 1-\max\left[\sup_y\left\{ \hat{G}_{n}(y,\delta,x) \right\},0\right].
\end{equation}
With respect to this choice of $\{\hat{\theta}^L(\cdot,\cdot),\hat{\theta}^U(\cdot,\cdot)\}$, Proposition \ref{prop:mainResConditionalGeneral} is a result under the most general conditions. Corollary \ref{cor:mainResDistrRegr} and Proposition \ref{prop:mainResRegrResids} give further concreteness for how exactly to guarantee the premise of Proposition \ref{prop:mainResConditionalGeneral} with respect to $\hat{G}(y,\delta,x)$. The idea behind the generic statement in Proposition \ref{prop:mainResConditionalGeneral} is to encourage extensions, especially those with the possibility of being more statistically efficient, with less restricted conditions than those in Proposition \ref{prop:mainResRegrResids}, or with modeling assumptions that are tailored to the application at hand.

\begin{proposition}[A model-free inequality\label{prop:mainResConditionalGeneral}]

If Assumption \ref{assump:strongCondIgno} holds, we have for all $\delta\in\mathbb{R}$ and all $x\in\text{support}(X_i)$:
\begin{equation}\label{eqn:keyIneqMainResX}
\left| \hat{\theta}^{L}(\delta,x)-\theta^{L}(\delta,x)\right| \vee\left|\hat{\theta}^{U}(\delta,x)-\theta^{U}(\delta,x)\right| \leq \sup_y\abs{ \hat{G}_n(y,\delta,x)-G(y,\delta,x) }.
\end{equation}

\end{proposition}

The implications of the deterministic inequality in \eqref{eqn:keyIneqMainResX} are interesting and perhaps a bit surprising. This inequality is stating that in a finite sample, the conditional bound estimators $\{\hat{\theta}^L(\delta,x),\hat{\theta}^U(\delta,x)\}$ in \eqref{eqn:plugInCondBoundL} and \eqref{eqn:plugInCondBoundU} are jointly no less accurate at estimating $\curl{\theta^L(\delta,x),\theta^U(\delta,x)}$ as the choice of $\hat{G}_n(y,\delta,x)$ is for $G(y,\delta,x)$, whatever the choice may be.

We can also turn \eqref{eqn:keyIneqMainResX} into a statement of frequentist confidence: if, additionally, $\hat{G}_n(y,\delta,x)$ is such that there exists a value $t_{\alpha}(\delta,x)$ such that:
\[
{\rm pr}\left\{ \sup_y\abs{ \hat{G}_n(y,\delta,x)-G(y,\delta,x) }\leq t_{\alpha}(\delta,x) \right\}\geq 1-\alpha,
\]
then we have that:
\[
{\rm pr}\left\{ \left|\hat{\theta}^{L}(\delta,x)-\theta^{L}(\delta,x)\right|\vee\left|\hat{\theta}^{U}(\delta,x)-\theta^{U}(\delta,x)\right|\leq t_{\alpha}(\delta, x) \right\}\geq 1-\alpha. 
\]
The term $t_{\alpha}(\delta,x)$ need not vary with $x$ or $\delta$; it can also be with respect to the concentration of $\hat{G}_n(y,\delta,x)$ uniformly across $x$ or across $\delta$ if more appropriate. Proposition \ref{prop:mainResRegrResids} below is an example with such a uniform guarantee.

With regard to specifying $\hat{G}_n(y,\delta,x)$, one choice is to plug in estimators of $F_{w}(y\mid x)$ ($w=0,1$) to arrive at a concentration inequality for the conditional bound estimators as Corollary \ref{cor:mainResDistrRegr} suggests. In doing so, provided the appropriate guarantee exists for the conditional $\cdf$ estimators, we actually get a strong guarantee for the bound estimators $\hat{\theta}^L(\delta,x)$ and $\hat{\theta}^U(\delta,x)$ that is simultaneous across all threshold values $\delta$ used to define $\pibt$ in Definition \ref{defn:propBenefit}.


\begin{corollary}[Concentration when plugging in conditional $\cdf$ estimators\label{cor:mainResDistrRegr}]

Let $\hat{F}_{wn}(y|x)$ denote an estimator for the conditional $\cdf$, $F_{wn}(y|x)$ ($w=0,1)$. If Assumption \ref{assump:strongCondIgno} holds, and the estimators $\hat{F}_{wn}(y|x)$ ($w=0,1)$ are such that there exists a value $t_{w,\alpha}(x)$ satisfying:
\[
{\rm pr}\curl{ \sup_y\abs{ \hat{F}_{wn}(y\mid x)-F_w(y\mid x) }\leq t_{w,\alpha}(x) }\geq 1-\alpha/2,
\]
then we have that:
\[
{\rm pr}\sqbrack{ \sup_\delta\curl{\left|\hat{\theta}^{L}(\delta,x)-\theta^{L}(\delta,x)\right|\vee\left|\hat{\theta}^{U}(\delta,x)-\theta^{U}(\delta,x)\right|}\leq \sum_{w=0,1}t_{w,\alpha}(x) }\geq 1-\alpha. 
\]

\end{corollary}


\subsection{More explicit conditional bounds with strategic use of regression residuals\label{sec:explicitBoundResids}}

The question now becomes how exactly to specify $\hat{F}_{wn}(y\mid x)$ in Corollary \ref{cor:mainResDistrRegr}, while guaranteeing the closeness between $\{\hat{\theta}^L(\delta,x),\hat{\theta}^U(\delta,x)\}$ and $\curl{\theta^L(\delta,x),\theta^U(\delta,x)}$ at some target confidence level. We explore one such choice using regression residuals for which such a high confidence guarantee is possible as summarized in Proposition \ref{prop:mainResRegrResids}. The motivation is that we would like something very similar to the plug-in estimator of ${\rm pr}\curl{Y_i(w)\leq y}$ given in \eqref{eqn:eCDFmarg} for the randomized experiment case. For the case that a locally linear approximation is desired, we give an actionable application of this result with regression residuals in \S~\ref{append:linearGuassianSetting}. 

Let $(X_1,W_1,Y_1(0),Y_1(1)),\dots,(X_n,W_n,Y_n(0),Y_n(1))$ be independent and identically distributed copies from a joint distribution. Denote the observable training data as:
\[
\Tcal:=\left\{ (X_i,W_i,Y_i) \right\}_{i=1}^n,
\]
where $Y_i=W_iY_i(1)+(1-W_i)Y_i(0)$. Consider partitioning $\Tcal$ into two independent splits $\Tcal_1$ and $\Tcal_2$. Denote the corresponding training indices as $\I_1,\I_2\subseteq\{1,\dots,n\}$ for $\Tcal_1$ and $\Tcal_2$, respectively. Denote $S_{w}:=\{i:\ W_i=w\}$, the index set of individuals in the sample in treatment group $w=0,1$. Let $n_w:=\abs{S_w\cap\I_2}$, the sample size in treatment group $w=0,1$ coming from data split $\Tcal_2$.

Further, denote
\[
\mu_w(x):=E\curl{Y_i(w)\mid X_i=x},
\]
the conditional expectation of the potential outcome as a function of the pre-treatment covariates. Denote the regression estimate using $\Tcal_1$ as $\hat{\mu}_w(x)$. Importantly, we will be able to reason about the counterfactual quantity $\mu_w(x)$ under Assumption \ref{assump:strongCondIgno}, because:
\[
\mu_w(x)= E\curl{Y_i\mid W_i=w,X_i=x},
\]
the conditional expectation of the observed outcome in treatment group $w=0,1$. Now consider the residuals:
\[
R_i(w):=Y_i(w)-\mu_w(X_i).
\]
Denote the approximation of $R_i(w)$ using $\hat{\mu}_{w}(\cdot)$ as
\[
\hat{R}_i(w):=Y_i(w)-\hat{\mu}_{w}(X_i).
\]

Motivated by the use of independent and identically distributed draws used to define $\hat{F}_{wn}(y)$ in \eqref{eqn:eCDFmarg} for the marginal $\pibt$ in a randomized experiment, we would like to approximate draws from the distribution 
\[
Y_i\mid X_i=x,W_i=w.
\]
With this in mind, we will specify:
\begin{equation}\label{eqn:condCDFResids}
\hat{F}_{wn}(y\mid x):=\frac{1}{n_w}\sum_{ i\in S_w\cap\I_2 }\indic{ \hat{\mu}_w(x)+\hat{R}_i(w) \leq y }.
\end{equation}
Considering that the definition of $R_i(w)$ means that $Y_i(w)=\mu_w(x)+R_i(w)$ conditional on $X_i=x$, it seems that using $\hat{\mu}_{w}(x)+\hat{R}_i(w)$ will make this choice of $\hat{F}_{wn}(y\mid x)$ a reasonable approximation to $F_{w}(y\mid x)$. 

Noting the liberal use of the plug-in principle en route to the choice of $\{\hat{\theta}^L(\cdot,\cdot),\hat{\theta}^U(\cdot,\cdot)\}$ using the conditional $\cdf$ estimators in \eqref{eqn:condCDFResids}, a concern is now what the regularity conditions must be so that overall the conditional bound estimators are close to their true values. For any given value of $x$, $\hat{F}_{wn}(y\mid x)$ is reusing residuals for indices in split $\I_2$ corresponding to subjects that are not necessarily in stratum $x$. Implicit in this use is that the distribution of $\{R_i(0),R_i(1)\}$ is the same across values of $x$. That is, we are using the independence assumption:
\[
\curl{R_i(0),R_i(1)}\indep X_i\ (i=1,\dots,n).
\]


Beyond this regularity condition, we also require that the distribution of $\hat{R}_i(w)$ approximates well the distribution of $R_i(w)$, which in turn requires that $\hat{\mu}_w(\cdot)$ be close to $\mu_w(\cdot)$. This explains the correction to the confidence level in Proposition \ref{prop:mainResRegrResids} with respect to how likely a deviation, in a uniform sense, is to occur between the true regression curve and the estimated regression curve based on random training data.

In Proposition \ref{prop:mainResRegrResids}, $\Xmat\in R^{|\I_1|\times p}$ is such that its rows are comprised of $X_i^T\in R^{1\times p}$ for $i\in\I_1$. 

\begin{proposition}[Justifying approach with regression residuals\label{prop:mainResRegrResids}]

For the bound estimators in \eqref{eqn:plugInCondBoundL} and \eqref{eqn:plugInCondBoundU}, let us specify $\hat{G}(y\mid x)=\hat{F}_{1n}(y\mid x)-\hat{F}_{0n}(y\mid x)$ using the conditional $\cdf$ estimators of \eqref{eqn:condCDFResids}. Let Assumption \ref{assump:strongCondIgno} hold, and assume further that the arbitrary joint distribution of $(R_i(0),R_i(1),X_i)$ is such that 
\[
(R_i(0),R_i(1))\indep X_i\ (i=1,\dots,n).
\]
Conditional on $\Xmat$, we have for any appropriate $0\leq\alpha\leq 1$:

\[\begin{aligned}
&\sup_{\delta,x}\curl{\left|\hat{\theta}^{L}(\delta,x)-\theta^{L}(\delta,x)\right|\vee\left|\hat{\theta}^{U}(\delta,x)-\theta^{U}(\delta,x)\right|}\\
\leq\ &\sum_{w=0,1}\sup_r\sqbrack{{\rm pr}\curl{r< R_i(w)\leq r+2t_w}\vee{\rm pr}\curl{r-2t_w<R_i(w)\leq r}}\\
&+\curl{\frac{\log(4/\alpha)}{2} }^{1/2}\parenth{n_0^{-\frac{1}{2}} + n_1^{-\frac{1}{2}} }
\end{aligned}\]
with probability at least 
\[
1-\alpha-\sum_{w=0,1}{\rm pr}\curl{ \sup_x\left|\hat{\mu}_{w}(x)-\mu_{w}(x)\right|>t_w \mid \Xmat}.
\]
Here, $t_0,t_1\geq0$ may depend on $\Xmat$. If they do not, we may remove the conditional statements. Appropriate values of $\alpha$ are those such that $1-\alpha-\sum_{w=0,1}{\rm pr}\curl{ \sup_x\left|\hat{\mu}_{w}(x)-\mu_{w}(x)\right|>t }$ is between $0$ and $1$. 
\end{proposition}

The proof of Proposition \ref{prop:mainResRegrResids} involves a strategic application of the Dvortetzky-Kiefer-Wolfowitz inequality \citep{DKWIneq1956,massartDKWIneq1990,NaamanDKWIneq2021} that is tailored to the imputed draws from the conditional potential outcome distribution. Consider the following corollary to Proposition \ref{prop:mainResRegrResids}. 

\begin{corollary}[Efficiency of Conditional Bound Estimators in Proposition \ref{prop:mainResRegrResids}\label{cor:effCondBounds}]

Let $\mathcal{F}$ be the function class containing our regression estimator, $\hat{\mu}_{w}(\cdot)$. Assume there exists a sequence $g_{n,\mathcal{F}}$, depending on $n$ and the complexity of $\mathcal{F}$, like feature dimension or regularisation parameters, such that
\[
\max_{w=0,1}\sup_x\left|\hat{\mu}_{w}(x)-\mu_{w}(x)\right|\lesssim g_{n,\mathcal{F}}
\]
with probability at least $1-\alpha$. Then:
\[
\sup_{\delta,x}\curl{\left|\hat{\theta}^{L}(\delta,x)-\theta^{L}(\delta,x)\right|\vee\left|\hat{\theta}^{U}(\delta,x)-\theta^{U}(\delta,x)\right|}\lesssim g_{n,\mathcal{F}}+\curl{\frac{\log(4/\alpha)}{2} }^{1/2}\parenth{n_0^{-\frac{1}{2}} + n_1^{-\frac{1}{2}} }
\]
holds with probability at least $1-2\alpha$. 
\end{corollary}


The rate of convergence of the regression estimator, $g_{n,\mathcal{F}}$, decreases down to $0$ as $n$ increases under proper specification. Should this rate satisfy
\[
\curl{\frac{\log(4/\alpha)}{2} }^{1/2}\parenth{n_0^{-\frac{1}{2}} + n_1^{-\frac{1}{2}} }\lesssim g_{n,\mathcal{F}},
\]
Corollary \ref{cor:effCondBounds} means that we lose nothing with respect to asymptotic efficiency with the plug-in estimators in \eqref{eqn:plugInCondBoundL} and \eqref{eqn:plugInCondBoundU} that build on an estimator of $\{\mu_0(\cdot),\mu_1(\cdot)\}$. Provided the framework gives $\hat{\mu}_w(x)$ ($w=0,1$), one can apply Proposition \ref{prop:mainResRegrResids} with any meta-learning algorithms that are used to estimate the conditional average treatment effect function ($\cate$) \citep{Kunzel4156,nieWager,grfAnnals,causalForests,kennedy2022OptimalDoublyRobustCATE,BurkhartRuiz2022}.

\subsubsection{An application to the linear Gaussian regression setup\label{append:linearGuassianSetting}}


We now explicitly apply Proposition \ref{prop:mainResRegrResids} to a case where we can derive non-asymptotic uniform guarantees on the regression function. It can be of interest when a locally linear approximation is desired, as made explicit by the following additional assumption. 

\begin{assumption}[Restricted data generating mechanism\label{assump:restrictedModel}]
We assume that 
\[
Y_i(w) = \mu_w(X_i)+R_i(w)\ (i=1,\dots,n).
\]
Here, $X_i\in\mathbb{R}^p$, $\mu_w(X_i)=\beta_w^T\Psi_w(X_i)$, and $\Psi_w:\mathbb{R}^p\to\mathbb{R}^d$ is a fixed mapping such that $\norm{\Psi_w(x)}_2\leq1$ and $d < n$. Moreover, $(R_i(0),R_i(1))$ ($i=1,\dots,n$) follow a joint distribution satisfying $(R_i(0),R_i(1))\indep X_i$ and such that marginally, \[
R_i(w)\sim\mathcal{N}(0,\sigma^2_w);\ \sigma^2_w>0\ (w=0,1).
\]
\end{assumption}

Denote the training data as:
\[
\Tcal:=\left\{ (X_i,W_i,Y_i) \right\}_{i=1}^n.
\]
Consider partitioning $\Tcal$ into two splits indexed by the disjoint sets $\I_1,\I_2\subseteq\{1,\dots,n\}$. Denote $S_{w}:=\{i:\ W_i=w\}$, the index set of individuals in the sample in treatment group $w=0,1$. Let $n_w:=\abs{S_w\cap\I_2}$, the sample size in treatment group $w=0,1$ coming from data split $\Tcal_2$. Now, let $\Psib_w\in\mathbb{R}^{ |\I_1\cap S_w|\times d}$ be such that its rows are made up by stacking $\Psi_w(X_i)^T$ for each $i\in\I_1\cap S_w$. Further, let $\Yvec_w\in\mathbb{R}^{|\I_1\cap S_w|\times 1}$ contain entries for the corresponding observed outcome $Y_i$ for each $i\in\I_1\cap S_w$. 

In Proposition \ref{prop:CIsGaussResids}, $\Phi$ denotes the $\cdf$ of the standard normal distribution. For a matrix $A\in\mathbb{R}^{m\times q}$, denote its operator norm as:
\[
\matnorm{A}_{op}:=\sup_{v\in\mathbb{R}^q:\ \norm{v}_2=1 }\norm{ Av }_2.
\]

\begin{proposition}\label{prop:CIsGaussResids}

Assume $\text{rank}(\Psib_w)=d$ ($w=0,1$) almost surely. Let $v_{d,\alpha}$ denote the $(1-\alpha/2)$th quantile for the $\chi^2_d$ distribution. Under Assumption \ref{assump:restrictedModel}, we have with confidence at least $(1-2\alpha)\times 100\%$ that uniformly across all pre-treatment covariate strata $x$,
\[
\theta(\delta,x)={\rm pr}\curl{Y_i(1)-Y_i(0)>\delta \mid X_i=x}
\]
is contained in the interval with starting point
\[\begin{aligned}
\hat{\theta}^{L}(\delta,x)&-\sum_{w=0,1}\Phi\curl{ \parenth{v_{d,\alpha}}^{1/2} \matnorm{ (\Psib_w^T\Psib_w)^{-1/2} }_{op}}\\
&+\sum_{w=0,1}\Phi\curl{ -\parenth{v_{d,\alpha}}^{1/2} \matnorm{ (\Psib_w^T\Psib_w)^{-1/2} }_{op} } \\
&-\curl{\frac{\log(4/\alpha)}{2} }^{1/2}\parenth{n_0^{-\frac{1}{2}} + n_1^{-\frac{1}{2}} }
\end{aligned}\]
and end point
\[\begin{aligned}
\hat{\theta}^{U}(\delta,x)&+\sum_{w=0,1}\Phi\left( \parenth{v_{d,\alpha}}^{1/2} \matnorm{ (\Psib_w^T\Psib_w)^{-1/2} }_{op} \right)\\
&-\sum_{w=0,1}\Phi\left( -\parenth{v_{d,\alpha}}^{1/2} \matnorm{ (\Psib_w^T\Psib_w)^{-1/2} }_{op} \right) \\
&+\curl{\frac{\log(4/\alpha)}{2} }^{1/2}\parenth{n_0^{-\frac{1}{2}} + n_1^{-\frac{1}{2}} }.
\end{aligned}\]

\end{proposition}

Although we do not prove that the term added and subtracted to the bound estimators $\{\hat{\theta}^L(\delta,x),\hat{\theta}^U(\delta,x)\}$ in Proposition \ref{prop:CIsGaussResids} goes to zero for a general feature matrix $(\Lambda_0,\Lambda_1)$ as sample size increase, the illustration of Proposition \ref{prop:CIsGaussResids} in Figure \ref{fig:powerAnalysisLmCase} of the appendix gives empirical evidence of this desirable decay. The proof of the inequality in Proposition \ref{prop:CIsGaussResids} is contained in Appendix \ref{append:proofCIsGaussResids}. It is an application of Proposition \ref{prop:mainResRegrResids} in the appendix, which specifies the bound estimators in terms of arbitrary mean-regression estimators $\{\hat{\mu}_0(x),\hat{\mu}_1(x)\}$. 

We believe Proposition \ref{prop:CIsGaussResids} may be regarded as partly representative of what can occur in practice while attempting to infer bounds on ${\rm pr}\curl{Y_i(1)-Y_i(0)>\delta \mid X_i=x}$ uniformly across $x$. The primary concerns for satisfactory margins of error are the typical concerns of regression: parsimony, multicollinearity, and feature dimension, all of which affect the convergence to zero of the terms involving $\Lambda_w$ $(w=0,1)$ in Proposition \ref{fig:powerAnalysisLmCase}.

\subsubsection{An example power analysis using Proposition \ref{prop:mainResRegrResids} and Proposition \ref{prop:CIsGaussResids}\label{append:linGaussCase}}

\begin{figure}[ht]
    \centering
    \includegraphics[width=\textwidth]{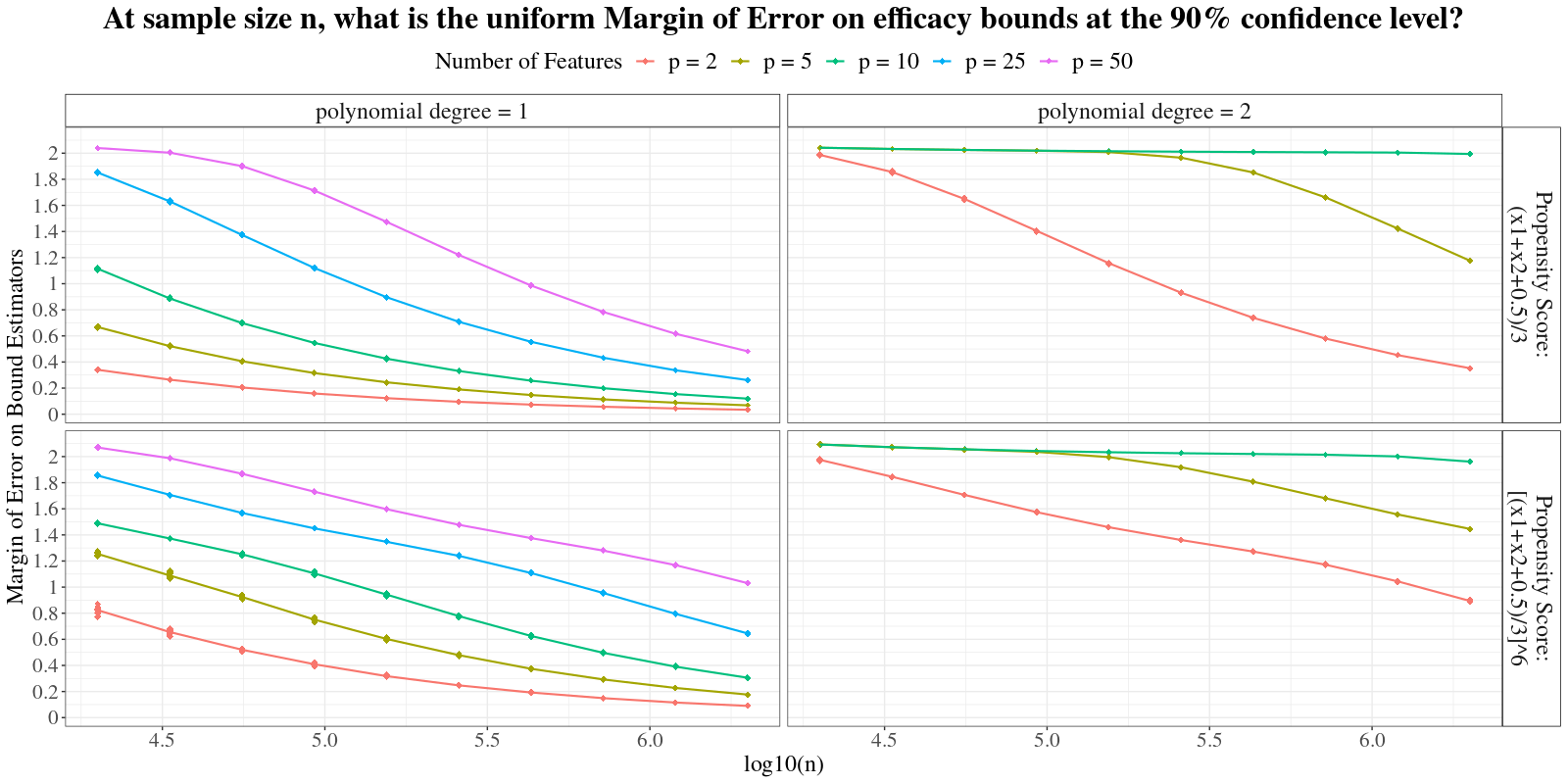}
    \caption{Power analysis based on Proposition \ref{prop:CIsGaussResids}. Each curve is the median calculated margin of error across 30 Monte Carlo simulations at the points that are also plotted. For these Monte Carlo draws, we generate $X_{ij}\sim\text{Uniform}(0,1)$ ($i=1,\dots,n$; $j=1,\dots,p$).}
    \label{fig:powerAnalysisLmCase}
\end{figure}

Under Assumption \ref{assump:restrictedModel}, Fig. \ref{fig:powerAnalysisLmCase} illustrates the behavior of the margin of error for 
\[
\sup_{\delta,x}\curl{\left|\hat{\theta}^{L}(\delta,x)-\theta^{L}(\delta,x)\right|\vee\left|\hat{\theta}^{U}(\delta,x)-\theta^{U}(\delta,x)\right|} 
\]
at the 90\% confidence level using the result in Proposition \ref{prop:CIsGaussResids}. As one can imagine, the distribution of $X_i$, the transformation $\Psi_w$ in Assumption \ref{assump:restrictedModel}, and the propensity score matter for an application of Proposition \ref{prop:CIsGaussResids}. That is, there may very well exist cases where the operator norm of $(\Psib_w^T\Psib_w)^{-1/2}$ does not decrease with $n$, along with cases of the propensity score where this norm decreases slowly due to insufficient treated or control units in the sample. To study this, the example summarized in Fig. \ref{fig:powerAnalysisLmCase} generates data as follows. 

We sample from a population in which each $X_{ij}\sim\text{Uniform}(0,1)$ across $i=1,\dots,n$ and $j=1,\dots,p$. The degree $q=1,2$ polynomial transformation of $X_i$ into $\Psi_w(X_i)$ includes all possible interaction terms of degree $1\leq k\leq q$, and each entry is re-scaled by $1/d^{1/2}$ so that $\norm{\Psi_w(X_i)}_2\leq1$ as in Assumption \ref{assump:restrictedModel}. For example, $\Psi_w(X_i)=(X_{i1},X_{i2},X_{i1}X_{i2},X_{i1}^2,X_{i2}^2)/{5^{1/2}}$ when $p=q=2$. The propensity score is $\prc{W_i=1}{X_i}=\curl{(X_{i1}+X_{i2}+0.5)/3}^m$\ $(m=1,6)$. The case with $m=1$ makes it so that the number of treated and control units are very close to each other in a random sample, while the case with $m=6$ will make it so that control units are typically much more represented. Moreover, the sample splitting is such that half of the observations are used to estimate $\{\mu_0(\cdot),\mu_1(\cdot)\}$, while the other half of the observations' residuals are used to estimate $\{{\theta}^L(\cdot,\cdot),\theta^U(\cdot,\cdot)\}$. 

\section{More on the scope of our results\label{sec:whenWorks}}
\subsection{Important concerns}
The contribution of this paper is in the quantification of the margin of error for the bounds on $\pibt$ given by the Makarov bounds \citep{makarovBounds,Frank1987BestpossibleBF,WILLIAMSON199089,fanSooPark2010}. Three important questions may come up regarding the scope of our results. For binary potential outcomes, is the use of the Makarov bounds on par with the use of the sharp Fr\'echet-Boole bounds \citep{boole1854,FrechetMaurice1935Gdtd,FrechetM.1960Sltd,hailperinBoole1986,muellerLiPearl2021PNSCovariates}? Moreover, can we have an alternative definition of benefiting from treatment in terms of the ratio of potential outcomes when the potential outcomes are strictly positive? Additionally, can we bound the proportion who are harmed by an intervention? The answer is yes to all three questions as we discuss in the next three subsections, respectively.




\subsection{When the potential outcomes are binary, the Makarov bounds on $\pibt$ are the same as the Boole-Fr\'echet bounds\label{sec:binaryOutcomes}}

Recall that the probability of necessity and sufficiency ($\pns$) for binary potential outcomes \citep{Pearl1999-sh,Tian2000-ul,pearlCausality} is the joint probability,
\[
{\rm pr}\curl{Y_i(1)=1,Y_i(0)=0}.
\]

For two measurable events $A$ and $B$, the Boole-Fr\'echet bounds on their joint probability are:
\[
\max\curl{\pr{A}+\pr{B}-1,0}\ \leq\ \pr{A\cap B}\ \leq\ \min\curl{\pr{A},\pr{B}}.
\]
These bounds are sharp \citep{Rschendorf1981SharpnessOfFrechetBoole}, meaning the bounds are attainable. For example, if $A\cap B=\emptyset$, then $\pr{A\cap B}=\max\curl{\pr{A}+\pr{B}-1,0}$ holds. Moreover, if $A\subseteq B$, then $\pr{A\cap B}=\min\curl{\pr{A},\pr{B}}$ is the case.

We claim in Proposition \ref{prop:boolFrechVsMakarov} that the Makarov bounds on $\pns$ are the same as the Boole-Fr\'echet bounds. From this proposition, it follows that all the estimation results we have shown are as applicable to the binary potential outcome case as any estimation approach involving the Boole-Fr\'echet bounds. 

\begin{proposition}[Boole-Fr\'echet bounds vs. Makarov bounds\label{prop:boolFrechVsMakarov}]

When the potential outcomes are binary, we have that the tightest Makarov bounds on $\pns$ are:
\[\sup_{0\leq\delta<1}\theta^L(\delta)\ \leq\ {\rm pr}\curl{Y_i(1)=1,Y_i(0)=0}\ \leq\ \inf_{0\leq\delta<1 }\theta^U(\delta)\]
in the marginal case, along with
\[\sup_{0\leq\delta<1}\theta^L(\delta,x)\ \leq\ {\rm pr}\curl{Y_i(1)=1,Y_i(0)=0 \mid X_i=x}\ \leq\ \inf_{0\leq\delta<1 }\theta^U(\delta,x)\]
in the conditional case. Moreover, these tightest Makarov bounds are the same as the Boole-Fr\'echet bounds:
\[\begin{aligned}
&\sup_{0\leq\delta<1}\theta^L(\delta)&=\ &\max\sqbrack{ {\rm pr}\curl{Y_i(1)=1}+{\rm pr}\curl{Y_i(0)=0}-1,0  }\\
&\inf_{0\leq\delta<1 }\theta^U(\delta)&=\ &\min\sqbrack{{\rm pr}\curl{Y_i(1)=1} ,{\rm pr}\curl{Y_i(0)=0}  }
\end{aligned}\]
and
\[\begin{aligned}
&\sup_{0\leq\delta<1}\theta^L(\delta,x)&=\ &\max\sqbrack{ {\rm pr}\curl{Y_i(1)=1 \mid X_i=x}+{\rm pr}\curl{Y_i(0)=0\mid X_i=x}-1,0  }\\
&\inf_{0\leq\delta<1 }\theta^U(\delta,x)&=\ &\min\sqbrack{{\rm pr}\curl{Y_i(1)=1 \mid X_i=x} ,{\rm pr}\curl{Y_i(0)=0\mid X_i=x}  }.
\end{aligned}\]

\end{proposition}

The proof of Proposition \ref{prop:boolFrechVsMakarov} can be found in Appendix \ref{append:proofPropBoolVsMak}. The general idea is to use that the $\cdf$s of binary potential outcomes have only three values, $0$, ${\rm pr}\curl{Y_i(w)=0}$, and $1$, across evaluation points $y\in\mathbb{R}$. We also use the fact that $\pibt$ is the same as $\pns$ when the threshold $\delta$ is in $[0,1)$.

\subsection{$\pibt$ in terms of the ratio of potential outcomes\label{sec:diffDefnRatio}}

We also have the following simple extension to ratios of potential outcomes. It is motivated by {Theorem 2} of \citet{WILLIAMSON199089}, which derives bounds for the $\cdf$ of a sum, difference, product, or ratio of two random variables with an unknown joint distribution. 

Suppose that we are interested in strictly positive potential outcomes $\tilde{Y}_i(0)$ and $\tilde{Y}_i(1)$. For example, this can be in a setting where the time to an event is the outcome of interest \citep{coxRegressionHazardRatio_1972,Stitelman2010_TTE_TMLE,Austin2014_TTE,Schober2018_TTE,Cai2020_TTE_TMLE}.
 Using a threshold $\tilde{\delta}>0$ (e.g. $\tilde{\delta}=1$), one may alternatively consider an individual to have benefited from treatment should the inequality
\[
\tilde{Y}_i(1)/\tilde{Y}_i(0)>\tilde{\delta}
\]
occur. That is, individual $i$ is deemed to have benefited from treatment should their treatment outcome be larger than their control outcome by a factor larger than $\tilde{\delta}$. Correspondingly, one may be interested in bounding the unidentifiable probability
\[
\tilde{\theta}(\tilde{\delta}):={\rm pr}\curl{\tilde{Y}_i(1)/\tilde{Y}_i(0)>\tilde{\delta}}. 
\]
To do so, consider the variable transformation
\[
Y_i(w)\ \dot{=}\ \log\curl{\tilde{Y}_i(w)}\ (w=0,1).
\]
Given our definition of $\theta(\delta)$ in Definition \ref{defn:propBenefit} and the one-to-one nature of the log transformation, we get that \[
\tilde{\theta}(\tilde{\delta})=\theta(\delta)
\]
when we set $\delta\ \dot{=}\ \log(\tilde{\delta})$. It follows that:
\[
\theta^L(\delta)\leq\tilde{\theta}(\tilde{\delta})\leq\theta^U(\delta). 
\]
We will simply have to work with the $\ecdf$ of $\log\{\tilde{Y}_i(w)\}$ when obtaining the estimators $\hat{\theta}^L(\delta)$ and $\hat{\theta}^U(\delta)$. Moreover, Proposition \ref{prop:mainResultMargPIBT} here is still useful to conduct inference on $\tilde{\theta}(\tilde{\delta})$ under the independent and identically distributed assumption. Likewise, the results in Section \ref{sec:condITEMakBnds} are also applicable when we work with the log-transformation.

\subsection{Reasoning about the proportion harmed by an intervention\label{sec:PIHI}}
Two days following the posting of our first manuscript on \texttt{arXiv}, a similar manuscript by \citet{kallus2022_propHarmed} was also posted on \texttt{arXiv}. This work studies the probability an individual is harmed by an intervention ($\pihi$) in the case of binary potential outcomes. Supposing instead that $Y_i=1$ is bad for an individual, while $Y_i=0$ is good for an individual, we have that the $\pihi$ is $\pns$:
\[
{\rm pr}\curl{Y_i(1)=1,Y_i(0)=0}.
\]
This is simply, but importantly, due to notational and vocabular semantics. Moreover, if the potential outcomes are real-valued, then one can refer to the quantity
\[
{\rm pr}\curl{Y_i(1)-Y_i(0)>\delta}
\]
as $\pihi$ instead of calling this quantity $\pibt$ as we do in Definition \ref{defn:propBenefit}. This holds analogously when we condition on $X_i=x$.

Given the discussion surrounding Proposition \ref{prop:boolFrechVsMakarov}, our model-free results extend to binary potential outcome cases for $\pihi$ as well. Given that the Boole-Fr\'echet inequalities underlay the theoretical estimation results for binary potential outcomes in \citet{kallus2022_propHarmed}, we believe that the contribution of their work compared to our present work is as follows. Most importantly, their results are for binary potential outcomes, while our results hold for binary, real-valued, or strictly positive potential outcomes. Next, their results include a doubly robust \citep{robins1994estimation,doublyRobust2011} estimation method which uses the estimated propensity score to adjust for covariates, possibly in an observational setting, when bounding marginal $\pihi$. This warrants further investigation for the case of non-binary potential outcomes. Moreover, their presentation of estimation theorems is in terms of big-O probability notation, i.e. statistical rates of convergence. We present results in terms of non-asymptotic concentration statements. The subtle difference is that our presentation helps provide nominal coverage guarantees for the confidence bands on $\pibt$ (or $\pihi$). Further, confidence bands for $\pihi$ presented in their work make use of a standard error: ${n}^{-\frac{1}{2}}$ times the standard deviation of a random variable being averaged to obtain the bound estimators. These standard-error-based confidence intervals, which are due to a central limit theorem, do not achieve nominal coverage until a fairly large sample size as demonstrated in their numerical results. See their {Algorithm 1} and {Fig. 3} for details. Finally, the non-asymptotic presentation of our results, having kept track of all constants, can also help with a statistical power analysis, as discussed in Sections \ref{sec:powerAnalysisMargPIBTCase} and \ref{append:linGaussCase}. However, keeping track of constants can prove difficult for some applications of Propositions \ref{prop:mainResConditionalGeneral} and \ref{prop:mainResRegrResids}. To extend the utility of Propositions \ref{prop:mainResConditionalGeneral} and \ref{prop:mainResRegrResids}, it may be interesting to study whether plausible bounds on these constants can be specified in such regression settings.

 \section{Proofs of the main theoretical results\label{append:proofsMainResults}}
\subsection{Overview}
The existing mathematical results we exploit include the following. Key to establishing the population-level target bounds on $\pibt$, we use the Makarov bounds first introduced in \citet{makarovBounds} and generalized in \citet{Frank1987BestpossibleBF} and \citet{WILLIAMSON199089}. These works establish a distribution-free bound on the cumulative distribution function ($\cdf$) on the sum (or difference or product) of two or more random variables having any unknown joint distribution and fixed marginal distributions. For the non-asymptotic concentration results (the margin of error derivations), the novel contribution of this paper, we use the Dvoretzky–Kiefer–Wolfowitz inequality \citep{DKWIneq1956,massartDKWIneq1990,NaamanDKWIneq2021}. We will refer to it as the DKW inequality for readability. The DKW inequality gives a non-parametric, non-asymptotic deviation inequality for the supremum difference at any evaluation point between a target $\cdf$ and its empirical analogue estimated with a sample of independent and identically distributed random variables. Importantly, \citet{massartDKWIneq1990} and \citet{NaamanDKWIneq2021} show that this inequality is tight under no additional assumptions. 
\subsection{Some key lemmas}

\begin{lemma}[Plug-in estimation of \citet{makarovBounds}'s conditional bounds\label{lem:condMakBoundEst}]

Consider jointly distributed random variables $(U_0,U_1,V)$. Denote:
\[
\gamma^L(\delta,v):= - \min\left[\inf_u\left\{\prc{U_1\leq u+\delta/2}{V=v}-\prc{U_0\leq u-\delta/2}{V=v}\right\},0\right]
\]
and
\[
\gamma^U(\delta,v):= 1- \max\left[\sup_u\left\{\prc{U_1\leq u+\delta/2}{V=v}-\prc{U_0\leq u-\delta/2}{V=v}\right\},0\right],
\]
the \citet{makarovBounds} and \citet{WILLIAMSON199089} lower and upper bounds for $\prc{U_1-U_0> \delta}{V=v}$. Denote 
\[
H(u,\delta,v):=\prc{U_1\leq u+\delta/2}{V=v}-\prc{U_0\leq u-\delta/2}{V=v}.
\]
Consider any estimator $\hat{H}(u,\delta,v)$ of $H(u,\delta,v)$ based on a sample
\[
\{(U_{i0},V_i)\}_{i=1}^{n_0}\cup\{(U_{i1},V_i)\}_{i=1}^{n_1}
\]
such that $(U_{i0},V_{i})$ ($i=1,\dots,n_0$) are independent and identically distributed copies of  $(U_0,V)$, while $(U_{i1},V_i)$ $(i=1,\dots,n_1)$ are independent and identically distributed copies of  $(U_1,V)$. Now let
\[
\hat{\gamma}^L(\delta,v):= - \min\left[\inf_u\left\{\hat{H}(u,\delta,v)\right\},0\right]
\]
and
\[
\hat{\gamma}^U(\delta,v):= 1-\max\left[\sup_u\left\{\hat{H}(u,\delta,v)\right\},0\right].
\]
We claim for every $\delta\in\mathbb{R}$ and every $v\in\text{support}(V)$ that:
\[
\hat{\gamma}^L(\delta,v)-\gamma^L(\delta,v)\ \leq\ \sup_u\curl{ H(u,\delta,v)-\hat{H}(u,\delta,v)+\abs{\hat{H}(u,\delta,v)-H(u,\delta,v)}}/2;
\]
\[
\gamma^U(\delta,v)-\hat{\gamma}^U(\delta,v)\ \leq\ \sup_u\curl{ \hat{H}(u,\delta,v)-H(u,\delta,v)+\abs{ \hat{H}(u,\delta,v)-H(u,\delta,v) }}/2;
\]
and
\[
\abs{ \hat{\gamma}^L(\delta,v)-\gamma^L(\delta,v) }\vee\abs{ \hat{\gamma}^U(\delta,v)-\gamma^U(\delta,v) }\ \leq\ \sup_u \abs{ \hat{H}(u,\delta,v)-H(u,\delta,v) }.
\]

\begin{proof}

For any real-valued function $g(t)$, denote its positive and negative parts, respectively, as:
\[
g^+(t)=-\max\left[g(t),0\right]\text{ and }g^-(t)= -\min\left[g(t),0\right].
\]
We have the following properties we will make use of:
\[
g(t)=g^+(t)-g^-(t);
\]
\[
|g(t)|=g^+(t) +g^-(t);
\]
\[
g^+(t)=\sqbrack{ |g(t)| + g(t) }/{2};
\]
\[
g^-(t)=\sqbrack{ |g(t)| - g(t) }/{2}.
\]

Below, we will use the positive and negative parts of 
\[
g_{\inf}(\delta,v):=\inf_u H(u,\delta,v)\text{ and }g_{\sup}(\delta,v):=\sup_u H(u,\delta,v)
\]
along with
\[
\hat{g}_{\inf}(\delta,v):=\inf_u\hat{H}(u,\delta,v)\text{ and }\hat{g}_{\sup}(\delta,v):=\sup_u\hat{H}(u,\delta,v).
\]

Consider:

For the lower bound on $\prc{U_1-U_0>\delta}{V=v}$, we are bounding the difference
\[\begin{aligned}
&\hat{\gamma}^{L}(\delta,v)-\gamma^{L}(\delta,v)\\
=\ &\hat{g}_{\inf}^-(\delta,v)-g_{\inf}^-(\delta,v)\\
\overset{(i)}{=}\ &\frac{1}{2}\sqbrack{ g_{\inf}(\delta,v)-\hat{g}_{\inf}(\delta,v) +\abs{\hat{g}_{\inf}(\delta,v)}-\abs{g_{\inf}(\delta,v)} }  \\
\overset{(ii)}{\leq}\ &\frac{1}{2}\sqbrack{ g_{\inf}(\delta,v)-\hat{g}_{\inf}(\delta,v) +\abs{ \hat{g}_{\inf}(\delta,v)-g_{\inf }(\delta,v) }}\\
=\ &\frac{1}{2}\sqbrack{\inf_uH(u,\delta,v)-\inf_u\hat{H}(u,\delta,v)  +\abs{\inf_u\hat{H}(u,\delta,v)-\inf_uH(u,\delta,v)}}.
\end{aligned}\]
In (i), we used the properties of the negative part of a function introduced above, while in (ii) we used reverse triangle inequality. Similarly, we have that:
\[
\abs{\gamma^{L}(\delta,v)-\hat{\gamma}^{L}(\delta,v)}\ \leq\ \abs{\inf_u\hat{H}(u,\delta,v)-\inf_uH(u,\delta,v)}.
\]

Consider that
\[\begin{aligned}
&\inf_u\hat{H}(u,\delta,v)-\inf_u H(u,\delta,v)\\
=\ &\inf_u\left\{\hat{H}(u,\delta,v)-H(u,\delta,v)+H(u,\delta,v)\right\}-\inf_u H(u,\delta,v)\\
\overset{(i)}{=}\ &\inf_u\left\{\hat{H}(u,\delta,v)-H(u,\delta,v)\right\}\\
\leq\ &\sup_u\left\{\hat{H}(u,\delta,v)-H(u,\delta,v)\right\}\\
\end{aligned}\]
Equality (i) follows from the fact that $\inf\{a+b:\ a\in A,b\in B\}=\inf(A)+\inf(B)$. Similarly,
\[\begin{aligned}
&\inf_u{H}(u,\delta,v)-\inf_u \hat{H}(u,\delta,v)\\
=\ &\inf_u\left\{\hat{H}(u,\delta,v)-H(u,\delta,v)+H(u,\delta,v)\right\}-\inf_u H(u,\delta,v)\\
\overset{\ }{=}\ &\inf_u\curl{H(u,\delta,v)-\hat{H}(u,\delta,v)}\\
\leq\ &\sup_u\curl{H(u,\delta,v)-\hat{H}(u,\delta,v)}\\
\end{aligned}\]

The previous two equations imply that

\[\begin{aligned}
\abs{\inf_u{H}(u,\delta,v)-\inf_u \hat{H}(u,\delta,v)}\leq\sup_u\abs{\hat{H}(u,\delta,v)-H(u,\delta,v)}.
\end{aligned}\]
So overall we have that
\begin{equation}\begin{aligned}\label{eqn:keyLemLwrPartNoAbs}
{\hat{\gamma}^{L}(\delta,v)-\gamma^{L}(\delta,v)}\leq\sup_u\sqbrack{ H(u,\delta,v)-\hat{H}(u,\delta,v)+\abs{\hat{H}(u,\delta,v)-H(u,\delta,v)}}/2\\
\end{aligned}\end{equation}
along with
\begin{equation}\begin{aligned}\label{eqn:keyLemLwrPart}
\left|\hat{\gamma}^{L}(\delta,v)-\gamma^{L}(\delta,v)\right|\leq \sup_u\abs{\hat{H}(u,\delta,v)-H(u,\delta,v)}.\\
\end{aligned}\end{equation}

Here, \eqref{eqn:keyLemLwrPartNoAbs} is the first desired conclusion.

For the upper bound on $\prc{ U_1-U_0>\delta }{V=v}$, we are bounding the difference
\[\begin{aligned}
&\gamma^{U}(\delta,v)-\hat{\gamma}^{U}(\delta,v)\\
=\ &\hat{g}_{\sup}^+(\delta,v)-g_{\sup}^+(\delta,v)\\
\overset{(i)}{=}\ &\frac{1}{2}\sqbrack{ \hat{g}_{\sup}(\delta,v)-g_{\sup}(\delta,v) +\abs{\hat{g}_{\sup}(\delta,v)}-\abs{g_{\sup}(\delta,v)} } \\
\overset{(ii)}{\leq}\ &\frac{1}{2}\sqbrack{ \hat{g}_{\sup}(\delta,v)-g_{\sup }(\delta,v)+\abs{ g_{\sup }(\delta,v)-\hat{g}_{\sup}(\delta,v) }}\\
=\ &\frac{1}{2}\curl{\sup_u \hat{H}(u,\delta,v)-\sup_u H(u,\delta,v)+\abs{\sup_u H(u,\delta,v)-\sup_u\hat{H}(u,\delta,v)}}.
\end{aligned}\]
In (i), we used the properties of the positive part of a function introduced above, while in (ii) we used triangle inequality followed by reverse triangle inequality. Noting that $\sup\{a+b:\ a\in A,b\in B\}=\sup(A)+\sup(B)$, we can arrive at the below equalities based on similar steps to the case with the lower bound:

\[\begin{aligned}
\sup_u{H}(u,\delta,v)-\sup_u \hat{H}(u,\delta,v)=\sup_u\curl{H(u,\delta,v)-\hat{H}(u,\delta,v)}
\end{aligned}\]
and
\[\begin{aligned}
\sup_u \hat{H}(u,\delta,v)-\sup_u{H}(u,\delta,v)=\sup_u\curl{\hat{H}(u,\delta,v)-H(u,\delta,v)}.
\end{aligned}\]

So overall we have that:
\begin{equation}\begin{aligned}\label{eqn:keyLemUprPartNoAbs}
\gamma^{U}(\delta,v)-\hat{\gamma}^{U}(\delta,v)\leq \sup_u\curl{\hat{H}(u,\delta,v)-H(u,\delta,v)+\abs{\hat{H}(u,\delta,v)-H(u,\delta,v)}}/2,\\
\end{aligned}\end{equation}
along with
\begin{equation}\begin{aligned}\label{eqn:keyLemUprPart}
\left|\hat{\gamma}^{U}(\delta,v)-\gamma^{U}(\delta,v)\right|\leq \sup_u\abs{\hat{H}(u,\delta,v)-H(u,\delta,v)},\\
\end{aligned}\end{equation}
as with the lower bound estimate. Here, \eqref{eqn:keyLemUprPartNoAbs} is the second desired conclusion.

The inequalities in \eqref{eqn:keyLemLwrPart} and \eqref{eqn:keyLemUprPart} give us the third desired conclusion:
\[\begin{aligned}
&\curl{\abs{ \hat{\gamma}^L(\delta,v)-\gamma^L(\delta,v) }\vee\abs{ \hat{\gamma}^U(\delta,v)-\gamma^U(\delta,v) }}&\leq\ & \sup_u \abs{ \hat{H}(u,\delta,v)-H(u,\delta,v) }.
\end{aligned}\]

\end{proof}
\end{lemma}

\begin{lemma}[Inequality for a sum of random variables\label{lem:bndProbOfSum}]

Let $U$ and $V$ be arbitrary real-valued random variables, and let $u,v\in\mathbb{R}$ be non-random scalars. We have that:
\[
\pr{U+V>u+v}\leq\pr{U>u}+\pr{V>v}.
\]
\begin{proof}

Consider that we have the following:
\[
\{U+V\leq u+v\}\supseteq \{U\leq u\}\cap\{V\leq v\}\iff\{U+V> u+v\}\subseteq \{U>u\}\cup\{V> v\}.
\]
This containment of events holds because $U\leq u$ and $V\leq v$ implies that $U+V\leq u+v$. It follows that 
\[
\pr{U+V>u+v}\leq \pr{\{U>u\}\cup\{V>v\}}\leq \pr{U>u}+\pr{V>v},
\]
where the second inequality is due to union bound.
\end{proof}
\end{lemma}

\begin{lemma}[Uniform estimation guarantee for the IPTW eCDF in Equation \eqref{eqn:eCDFmarg}\label{lem:eCDFGuarantee}]\ \\

We have that 
\[
\pr{\sup_y\abs{ \hat{F}_w(y|x)-F_w(y|x) }\ \leq\  n_{\ell(x )}^{-1/2} C_{w,\beta}(x) } \geq 1-\beta.
\]
Here,
\[
C_{w,\beta}(x):=\begin{cases}
[2^{-1}n_{\ell(x )} n_{w\ell(x)}^{-1}\log(2/\beta)]^{1/2}&\text{ if }e(x')=n_{1\ell(x)}/n_{\ell(x )}\text{ for each }x'\in\xcal_{\ell(x)}\\
\frac{2\sqbrack{ 2\log^{\frac{1}{2}}\parenth{ n_{\ell(x)}+1 }+\log^{\frac{1}{2}}(1/\beta)}}{ n_{\ellx} ^{1/2}\parenth{ w\underbar{e}+(1-w)(1-\bar{e}) } }&\text{ otherwise }
\end{cases}.
\]

\begin{proof}\ \\
\textbf{\underline{Case 1: randomized experiment}}\\

For the case the $e(x')=n_{1\ell(x)}/n_{\ell(x )}$ for each $x'\in\xcal_{\ellx}$, the estimator in \eqref{eqn:eCDFmarg} reduces to:
\[
\frac{1}{n_{w\ell(x)}}\sum_{i:\ W_i=w,\ellxi=\ellx}\indic{Y_i\leq y}.
\]
When $(Y_1,W_1,X_1),\dots,(Y_n,W_n,X_n)$ are iid copies to begin with, it follows that $Y_i$ is iid among indices such that $W_i=w$ and $\ellxi=\ellx$. So the iid assumption of the Dvoretzky-Kiefer-Wolfowitz-Massart (DKWM) inequality holds \citep{DKWIneq1956,massartDKWIneq1990,NaamanDKWIneq2021,overlap2021}. This gives us the desired conclusion.

\textbf{\underline{Case 2: general observational setting}}\\

Let $t> 0$. Let $S\in\curl{-1,1}$. Consider that:

\begin{equation}\begin{aligned}\label{eqn:funcHoeff}
&\sup_{y}\abs{\hat{F}_{w}(y\mid x)-F_w(y\mid x) }\\
=\ &\sup_{y,S}S\curl{\frac{1}{n_{\ell(x )} }\sum_{ i:\ \ellxi=\ellx }\parenth{\frac{\indic{Y_i\leq y}\indic{W_i=w}}{ we(X_i)+(1-w)\sqbrack{1-e(X_i)} }-\EC{\frac{\indic{Y_i\leq y}\indic{W_i=w}}{ we(X_i)+(1-w)\sqbrack{1-e(X_i)} } }{\ell(X_i)=\ell(x)} } }\\
\overset{(i)}{\leq}\ &t+\E{ \sup_{y,S}S\curl{\frac{1}{n_{\ell(x )} }\sum_{ i:\ \ellxi=\ellx }\parenth{\frac{\indic{Y_i\leq y}\indic{W_i=w}}{ we(X_i)+(1-w)\sqbrack{1-e(X_i)} }-\EC{\frac{\indic{Y_i\leq y}\indic{W_i=w}}{ we(X_i)+(1-w)\sqbrack{1-e(X_i)} }}{\ell(X_i)=\ell(x)} } } }\\
\overset{(ii)}{\leq}\ & t+\frac{4\log^{\frac{1}{2} }(n_{\ell(x)}+1) }{ n_{\ell(x)}^{\frac{1}{2} }\parenth{w\underbar{e}+(1-w)(1-\bar{e})} }.\\
\end{aligned}\end{equation}

Here, inequality (i) holds with probability at least 
\[1-\exp\parenth{ -\frac{n_{\ellx} t^2 \sqbrack{ w\underbar{e}+(1-w)(1-\bar{e}) }^2 }{4} }.\] 
It is by Functional Hoeffding's inequality \citep{wainwright_2019}; its premise holds: with probability 1,
\[
0\ \leq\ \sup_{y,S } S\curl{\frac{\indic{Y_i\leq y}\indic{W_i=w}}{ we(X_i)+(1-w)\sqbrack{1-e(X_i)} }-\EC{\frac{\indic{Y_i\leq y}\indic{W_i=w}}{ we(X_i)+(1-w)\sqbrack{1-e(X_i)} }}{\ell(X_i)=\ell(x)} }\leq\ \frac{1}{w\underbar{e}+(1-w)(1-\bar{e}) }.
\]
Meanwhile, inequality (ii) holds with probability $1$ due to Lemma \ref{lem:expectedDevECDF}. 

With some algebra, we see that:
\[
\pr{ \sup_{y}\abs{\hat{F}_{w}(y)-F_w(y) }\leq  \frac{2\sqbrack{ 2\log^{\frac{1}{2}}\parenth{ n_{\ell(x)}+1 }+\log^{\frac{1}{2}}(1/\beta)}}{ n_{\ellx} ^{1/2}\parenth{ w\underbar{e}+(1-w)(1-\bar{e}) } } }\geq 1-\beta.
\]

\end{proof}

\end{lemma}

\subsubsection{Inferring an Off Policy Outcome Distribution}

Suppose the following generation process for our observed tuple of interest:
\[
(Y_i,W_i,X_i)\iid p(y|w,x)p(w|x)p(x);\ i=1,\dots,n.
\]
We are interested in inferring the outcome had we generated treatment with policy $\pi(w|x)$ rather than $p(w|x)$. Here, $\pi$ is understood to be a conditional probability density (pdf) or probability mass function (pmf) for $W_i$--a vector, discrete, and/or continuous treatment variable. We define the generative process for such a random Variable in Definition \ref{defn:opeOutcome}.

\begin{definition}[Off-Policy Outcome\label{defn:opeOutcome}]
Consider the generative process to obtain our outcome of interest,
\[
Y_i(\pi),
\]
under treatment assignment policy $\pi$:
\begin{enumerate}
    \item $X_i\sim p(x)$.
    \item Given $X_i=x$, $W_i\sim \pi(w|x)$.
    \item Given $X_i=x$ and $W_i=w$, $Y_i\sim p(y|w,x)$.
\end{enumerate}

\end{definition}

The potential random variable $Y_i(\pi)$ of Definition \ref{defn:opeOutcome} is exactly $Y_i$ when $\pi(w|x)=p(w|x)$. Additionally, when $W_i$ is discrete and $\pi(w|x)=1$ for a fixed value $w$, while
$\pi(w'|x)=0$ for $w'\neq w$, then 
\[
Y_i(\pi)\dot{=}Y_i(w)
\]
is exactly what is known as a potential outcome \citep{imbens_rubin_2015} or counterfactual \citep{pearlCausality,hernan2020}. This instance is the random variable we consider in the main text, but Definition \ref{defn:opeOutcome} is more general.

Denote the conditional CDF of $Y_i(\pi)$ given $\ell(X_i)=\ell(x)$ as 
\[
F_\pi(y|x):= \prc{Y_i(\pi)\leq y}{\ell(X_i)=\ell(x)}.
\]
Here, $\ell(x)=l$ if $x\in\xcal_l$, where
\[
\xcal_l\subseteq \xcal:=\text{support}(X_i)\ (l=1,\dots,L;\ L\geq 1).
\]

Our estimator is:
\begin{equation}\label{eqn:opeCDFEstimator}
\hat{F}_\pi(y|x):=\frac{1}{n_{\ell(x)}}\sum_{i:\ \ell(X_i)=\ell(x)}\frac{\indic{Y_i\leq y}\pi(W_i|X_i) }{ p(W_i|X_i) }.
\end{equation}

$\hat{F}_\pi(y|x)$ in \eqref{eqn:opeCDFEstimator} is unbiased for $F_\pi(y|x)$:
\begin{equation}\begin{aligned}
\label{eq7}
&\EC{\hat{F}_\pi(t|x)}{\ell(X_i)=l;i\ : \ell(X_i)=\ell(x)}\\
&=\ \frac{1}{n_{l}}\sum_{i:\ \ell(X_i)=l}\int_{(y,w,x)}\frac{\indic{Y_i\leq t}\pi(W_i|X_i) }{ p(W_i|X_i) }p(y,w,x|\ell(x)=l)\\
&=\ \frac{1}{n_{l}}\sum_{i:\ \ell(X_i)=l}\int_x\int_w\int_y \indic{y\leq t}p(y|w,x)\frac{\pi(w|x)}{p(w|x)}p(w|x) p(x|\ell(x)=l) \\
&=\ \frac{1}{n_{l}}\sum_{i:\ \ell(X_i)=l}\int_x\int_w\int_y \indic{y\leq t}p(y|w,x)\pi(w|x) p(x|\ell(x)=l) \\
&=\ \frac{1}{n_{l}}\sum_{i:\ \ell(X_i)=l}\int_x\int_{w} \prc{Y_i\leq t}{X_i=x,W_i=w}\pi(w|x) p(x|\ell(x)=l) \\
&=\ \frac{1}{n_{l}}\sum_{i:\ \ell(X_i)=l}\int_x \prc{Y_i(\pi)\leq t}{X_i=x}p(x|\ell(x)=l) \\
&=\ \frac{1}{n_{l}}\sum_{i:\ \ell(X_i)=l} \prc{Y_i(\pi)\leq t}{\ell(X_i)=l}\\
=\ & F_\pi(t|x).
\end{aligned}\end{equation}

In Lemma \ref{lem:expectedDevECDF}, we now study the finite sample performance of the Off-policy estimator in \eqref{eqn:opeCDFEstimator} across possible evaluation points $y$.

\begin{lemma}[Bounding the expected uniform deviation of the eCDF estimator\label{lem:expectedDevECDF}]\ \\

We have that

\[\begin{aligned}
&\EC{\sup_{y}\abs{\hat{F}_\pi(y|x)-F_\pi(y|x) }}{\ell(X_i)=\ell(x)}\\
\overset{}{\leq}\ & 4\norm{\frac{\pi(W_i|X_i)}{ p(W_i|X_i)  } }_{x,\infty}\log^{\frac{1}{2} }(n_{\ell(x)}+1) n_{\ell(x)}^{-\frac{1}{2} }.\
\end{aligned}\]
Here, $\norm{V_i}_{x,\infty}:=\inf\curl{t:\ \prc{V_i\leq t}{\ell(X_i)=\ell(x)}=1}$.

\begin{proof}\ \\

Let $(Y_i,X_i,W_i)\iid (Y_i^*,X_i^*,W_i^*)$ for $i=1,\dots,n$. Independent of these, let $\sigma_i\iid Rademacher$. Below, also let $S\in\curl{-1,1}$. Omitting the dependence on $\ell(X_i)=\ell(x)$, as it is a straightforward extension, we have:
\[\begin{aligned}
&\E{\sup_{y}\abs{\frac{1}{n}\sum_{i=1}^n\parenth{\frac{\indic{Y_i\leq y}\pi(W_i|X_i)}{ p(W_i|X_i) }-\E{\frac{\indic{Y_i\leq y}\pi(W_i|X_i)}{ p(W_i|X_i) }}} }}\\
=\ &\E{\sup_{y,S}S \curl{\frac{1}{n}\sum_{i=1}^n\parenth{\frac{\indic{Y_i\leq y}\pi(W_i|X_i)}{ p(W_i|X_i) }-\E{\frac{\indic{Y_i\leq y}\pi(W_i|X_i)}{ p(W_i|X_i) }}} }}\\
=\ &\E{\sup_{y,S}S \curl{\frac{1}{n}\sum_{i=1}^n\parenth{\frac{\indic{Y_i\leq y}\pi(W_i|X_i)}{ p(W_i|X_i) }-\E{\frac{\indic{Y_i^*\leq y}\pi(W_i^*|X_i^*)}{ p(W_i^*|X_i^*) }}} }}\\
\leq\ &\E{\sup_{y } \abs{\frac{1}{n}\sum_{i=1}^n\parenth{\frac{\indic{Y_i\leq y}\pi(W_i|X_i)}{ p(W_i|X_i) }-\E{\frac{\indic{Y_i^*\leq y}\pi(W_i^*|X_i^*)}{ p(W_i^*|X_i^*) }}} }}\\
\overset{}{\leq}\ & \E{\sup_{y } \abs{\frac{1}{n}\sum_{i=1}^n\parenth{\frac{\indic{Y_i\leq y}\pi(W_i|X_i)}{ p(W_i|X_i) }-\frac{\indic{Y_i^*\leq y}\pi(W_i^*|X_i^*)}{ p(W_i^*|X_i^*) }}} }\\
\overset{}{=}\ & \E{\sup_{y } \abs{\frac{1}{n}\sum_{i=1}^n\sigma_i \parenth{\frac{\indic{Y_i\leq y}\pi(W_i|X_i)}{ p(W_i|X_i) }-\frac{\indic{Y_i^*\leq y}\pi(W_i^*|X_i^*)}{ p(W_i^*|X_i^*) }}} }\\
\overset{}{\leq}\ & 2\E{\sup_{y } \abs{\frac{1}{n}\sum_{i=1}^n\sigma_i \frac{\indic{Y_i\leq y}\pi(W_i|X_i)}{ p(W_i|X_i) } } }\\
\overset{(i)}{\leq}\ & 4\norm{\frac{\pi(W_i|X_i)}{ p(W_i|X_i)  } }_{\infty}\log^{\frac{1}{2} }(n+1) n^{-\frac{1}{2} }.\\
\end{aligned}\]

Up to right before inequality (i), we follow the standard symmetrization argument for bounding the expectation of an empirical process in terms of what is called the Rademacher complexity of a function class that indexes the empirical process. See for example,  \citet{wainwright_2019} and \citet{vanderVaart1996}.

Inequality (i) follows due to \textit{Lemma 4.14} in \citet{wainwright_2019}. This lemma covers what is called the Rademacher complexity, an expectation of the across independent random variables $(T_1,\dots,T_n)$ and $(\sigma_1,\dots,\sigma_n)$:
\[
\E{\sup_{g\in\mathcal{G}}\abs{\frac{1}{n}\sum_{i=1}^n\sigma_i g(T_i) }}.
\]
The Lemma particularly covers the class of functions $\mathcal{G}$ satisfying \eqref{eqn:boundedPolyDescr}. Denote the evaluation of these functions on a finite set of points as: 
\[
\mathcal{G}(t_1,\dots,t_n)=\curl{g(t):\ t\in\curl{t_1,\dots,t_n},g\in\mathcal{G}}.
\]
\textit{Lemma 4.14} in \citet{wainwright_2019} requires the existence of some $\nu>0$ such that the cardinality of $\mathcal{G}(t_1,\dots,t_n)$ is bounded by a polynomial as:
\begin{equation}\label{eqn:boundedPolyDescr}
\abs{\mathcal{G}(t_1,\dots,t_n)}\leq \parenth{n+1}^\nu,
\end{equation}
a condition known as "Polynomial Descrimination."

In our case, the random variable $T_i=(Y_i,W_i,X_i)$ and our infinitely sized function class, indexed by $y$ and evaluated at $t=(a,b,c)$, is:
\[
\mathcal{G}=\curl{ \frac{\indic{a\leq y}\pi(b|c)}{ p(b|c) }:\ y\in\mathbb{R}  }.
\]
Consider that for any $t=(a,b,c)$ and any $y$, it must be the case that 
\[
\frac{\indic{a\leq y}\pi(b|c)}{ p(b|c) }\in\curl{0,\frac{\pi(b|c)}{ p(b|c) }}.
\]
Thus, \eqref{eqn:boundedPolyDescr} is satisfied with $\nu=1$.

\end{proof}
\end{lemma}

\subsection{The proof of Proposition \ref{prop:mainResultMargPIBT} \label{append:proofMainResMargPIBT}}

\begin{proof}[Proof of Proposition \ref{prop:mainResultMargPIBT}]

Let $x\in\xcal$ be fixed. We apply Lemma \ref{lem:condMakBoundEst} with $U_0:=Y_i(0)$, $U_1:=Y_i(1)$, and $V=\ell(X_i)$, which gives:
\begin{equation}\begin{aligned}\label{eqn:ineqMargCaseProof}
&\sup_\delta \curl{\abs{ \hat{\theta}^L(\delta,x)-\theta^L(\delta,x) }\vee\abs{ \hat{\theta}^U(\delta,x)-\theta^U(\delta,x) }}\\
\leq\ & \sup_{\delta,y} \abs{ \curl{\hat{F}_{1n}(y+\delta/2\mid x)-\hat{F}_{0n}(y-\delta/2\mid x )}-\curl{{F}_{1}(y+\delta/2 \mid x )-{F}_{0}(y-\delta/2 \mid x) } }\\
\leq\ &\sup_{\delta,y} \abs{ \hat{F}_{1n}(y+\delta/2 \mid x )-{F}_{1}(y+\delta/2 \mid x ) }+\sup_{\delta,y} \abs{ \hat{F}_{0n}(y-\delta/2 \mid x )-{F}_{0}(y-\delta/2 \mid x ) }\\
=\ &\sup_{y} \abs{ \hat{F}_{1n}(y\mid x)-{F}_{1}(y \mid x)  }+\sup_{y} \abs{ \hat{F}_{0n}(y\mid x)-{F}_{0}(y\mid x) }\\
\overset{(i)}{\leq}\ &n_{\ell(x )}^{-1/2}\sum_{w=0,1} C_{w,\beta/2}(x).
\end{aligned}\end{equation}

Here, inequality (i) holds with probability at least $1-\beta$ due to an application of Lemma \ref{lem:eCDFGuarantee} along with Lemma \ref{lem:bndProbOfSum}.

So far, $x$ is fixed. We now consider the two cases as stated in the proposition.

\textbf{Case 1:} Consider evaluating the bound estimators at a random new test point, $\xnp$, independent of our training data. We have:
\[\begin{aligned}
&\pr{ \sup_\delta \curl{\abs{ \hat{\theta}^L(\delta,\xnp)-\theta^L(\delta,\xnp) }\vee\abs{ \hat{\theta}^U(\delta,\xnp)-\theta^U(\delta,\xnp) }} \leq n_{\ell(\xnp )}^{-1/2}\sum_{w=0,1} C_{w,\alpha/2}(\xnp) }\\
\overset{(i)}{=}\ &\sum_{ l=1}^L\pr{\ell(\xnp)=l}\\
&\times \pr{ \sup_\delta \curl{\abs{ \hat{\theta}^L(\delta,x_l)-\theta^L(\delta,x_l) }\vee\abs{ \hat{\theta}^U(\delta,x_l)-\theta^U(\delta,x_l) }} \leq n_{\ell(x_l)}^{-1/2}\sum_{w=0,1} C_{w,\alpha/2}(x_l) }\\
\overset{(ii)}{\geq}\ &1-\alpha.
\end{aligned}\]

Equality (i) holds by law of total probability. Here, we insert arbitrary fixed points $x_l\in\xcal_l$ ($l=1,\dots,L$). The conditioning statement in the second factor inside the sum is omitted since $\xnp$ is independent of the bound estimators after we insert the point $x_l$. Next, inequality (ii) holds by applying \eqref{eqn:ineqMargCaseProof} at each point $x_l$ with $\beta=\alpha/2$.

\textbf{Case 2:} Now consider a uniform guarantee on the bound estimators across $x\in\xcal$, in addition to across $\delta\in\mathbb{R}$.

We have:
\[\begin{aligned}
&\pr{ \sup_{\delta,x} \curl{\abs{ \hat{\theta}^L(\delta,x)-\theta^L(\delta,x) }\vee\abs{ \hat{\theta}^U(\delta,x)-\theta^U(\delta,x) }} \leq t }\\
\geq\ &1-\alpha
\end{aligned}\]

by union bound provided we set
\[
t\ \dot{=}\ \max_l \curl{n_l^{-1/2}\sum_{w=0,1} C_{w,\alpha/(2L)}(x_l)}.
\]

\end{proof}

\subsection{The proof of Proposition \ref{prop:mainResConditionalGeneral}\label{append:proofMainResCondGen}}

\begin{proof}[Proof of Proposition \ref{prop:mainResConditionalGeneral}]

This result is a direct consequence of Lemma \ref{lem:condMakBoundEst} with $U_{i0}:=Y_i(0)$, $U_{i1}:=Y_i(1)$, and $V_i:=X_i$. 
\end{proof}

\subsection{The proof of Corollary \ref{cor:mainResDistrRegr}\label{append:ProofMainResDistrRegr}}
\begin{proof}[Proof of Corollary \ref{cor:mainResDistrRegr}]

Specify
\[
\hat{G}(y,\delta,x) = \hat{F}_{1n}(y+\delta/2\mid x)-\hat{F}_{0n}(y-\delta/2\mid x),
\]
which simply plugs in the estimate of $F_{1}(y+\delta/2\mid x)$ and $F_{0}(y-\delta/2\mid x)$. Consider that:

\[
\sup_y\abs{ \hat{F}_{wn}(y\mid x)-F_w(y\mid x) }=\sup_{a,y}\abs{ \hat{F}_{wn}(y+a\mid x)-F_w(y+a\mid x) },
\]
where the supremum is with respect to $y$ is across the real line, while the supremum across $(y,a)$ is across the euclidean plane. Therefore, when the inequality
\begin{equation}\label{eqn:distrRegrIneq}
\sup_y\abs{ \hat{F}_{wn}(y\mid x)-F_w(y\mid x) }\leq t_{w,\alpha}(x)
\end{equation}
holds, triangle inequality and the definition of $G(y,\delta,x)$ tell us that
\begin{equation}\begin{aligned}\label{eqn:distrTrtmntEffIneq}
&\sup_{\delta,y}\abs{ \hat{G}_{n}(y,\delta,x)-G(y,\delta,x) }&\leq\ &\sup_{\delta,y}\abs{ \hat{F}_{1n}(y+\delta/2\mid x)-F_w(y+\delta/2\mid x) }\\
&&&+\sup_{\delta,y}\abs{ \hat{F}_{0n}(y-\delta/2\mid x)-F_0(y-\delta/2\mid x) }\\
&&=\ &\sup_{y}\abs{ \hat{F}_{1n}(y\mid x)-F_w(y\mid x) }+\sup_{y}\abs{ \hat{F}_{0n}(y\mid x)-F_0(y\mid x) }\\
&&\leq\ &\sum_{w=0,1}t_{w,\alpha}(x).
\end{aligned}\end{equation}
Combining inequality \eqref{eqn:distrTrtmntEffIneq} and inequality \eqref{eqn:keyIneqMainResX} in Proposition \ref{prop:mainResConditionalGeneral} (after taking the supremum on both sides with respect $\delta$), we get:

\begin{equation}\label{eqn:corDevBndsConfLev}
\sup_\delta\curl{\left|\hat{\theta}^{L}(\delta,x)-\theta^{L}(\delta,x)\right|\vee\left|\hat{\theta}^{U}(\delta,x)-\theta^{U}(\delta,x)\right|}\leq\sum_{w=0,1} t_{w,\alpha}(x),
\end{equation}
also. Now, if the inequality in \eqref{eqn:distrRegrIneq} holds with high probability, which is the premise of this corollary, then so must \eqref{eqn:corDevBndsConfLev}. 
\end{proof}

\subsection{The proof of Proposition \ref{prop:mainResRegrResids}\label{append:proofResidApproach}}

\begin{proof}[Proof of Proposition \ref{prop:mainResRegrResids}]

Given inequality \eqref{eqn:keyIneqMainResX} in Proposition \ref{prop:mainResConditionalGeneral}, our specification for 
\[
\hat{G}(y,\delta,x):=\hat{F}_{1n}(y\mid x)-\hat{F}_{0n}(y\mid x),
\]
and a similar argument to the proof of Corollary \ref{cor:mainResDistrRegr} in Appendix \ref{append:ProofMainResDistrRegr}, we have that:

\begin{equation}\label{eqn:KeyIneqResidCase}
\sup_x\curl{\left|\hat{\theta}^{L}(\delta,x)-\theta^{L}(\delta,x)\right|\vee\left|\hat{\theta}^{U}(\delta,x)-\theta^{U}(\delta,x)\right|}\leq \sum_{w=0,1}\sup_{y,x}\abs{ \hat{F}_{wn}(y\mid x)-F_{w}(y\mid x) }.
\end{equation}
Consider $s_0,s_1\geq 0$. Due to Lemma \ref{lem:bndProbOfSum} and \eqref{eqn:KeyIneqResidCase}, we have that:
\begin{equation}\begin{aligned}\label{eqn:Key2IneqResidCase}
&{\rm pr}\sqbrack{ \sup_x\curl{\left|\hat{\theta}^{L}(\delta,x)-\theta^{L}(\delta,x)\right|\vee\left|\hat{\theta}^{U}(\delta,x)-\theta^{U}(\delta,x)\right|}> s_0+s_1 \mid \Xmat }\\
\leq\ &{\rm pr}\curl{ \sum_{w=0,1}\sup_{y,x}\abs{ \hat{F}_{wn}(y\mid x)-F_{w}(y\mid x) }>s_0+s_1 \mid \Xmat}\\
\leq\ &\sum_{w=0,1}{\rm pr}\curl{ \sup_{y,x}\abs{ \hat{F}_{wn}(y\mid x)-F_{w}(y\mid x) }>s_w \mid \Xmat}.
\end{aligned}\end{equation}
Let $t_w\geq 0$. Lemma \ref{lem:condCDFEstResids} tells us that for 
\[\begin{aligned}
&s_w&\dot{=}\ &\sup_r  \sqbrack{{\rm pr}\curl{ r<R_i(w)\leq r+2t_w \mid \Xmat} \vee {\rm pr}\curl{ r-2t_w<R_i(w)\leq r \mid \Xmat} }\\
&&&+\curl{\frac{\log(4/\alpha)}{2} }^{1/2} n_w^{-\frac{1}{2}},
\end{aligned}\]
we get:
\begin{equation}\label{eqn:Key3IneqResidCase}
{\rm pr}\curl{ \sup_{y,x}\abs{ \hat{F}_{wn}(y\mid x)-F_{w}(y\mid x) }>s_w \mid \Xmat}\leq \alpha/2+{\rm pr}\curl{ \sup_x\abs{ \hat{\mu}_w(x)-\mu_w(x) } > t_w \mid \Xmat }.
\end{equation}
Now, combining \eqref{eqn:Key2IneqResidCase} and \eqref{eqn:Key3IneqResidCase}, we have conditional on $\Xmat$:
\[\begin{aligned}
&{\rm pr}\sqbrack{ \sup_x\curl{\left|\hat{\theta}^{L}(\delta,x)-\theta^{L}(\delta,x)\right|\vee\left|\hat{\theta}^{U}(\delta,x)-\theta^{U}(\delta,x)\right|}> s_0+s_1 \mid \Xmat }\\
\leq\ &\alpha+\sum_{w=0,1}{\rm pr}\curl{ \sup_x\abs{ \hat{\mu}_w(x)-\mu_w(x) } > t_w \mid \Xmat },\\
\end{aligned}\]
which is the desired conclusion.

\end{proof}

\subsection{A lemma for Proposition \ref{prop:mainResRegrResids}: conditional $\cdf$ estimation with regression residuals}

\begin{lemma}[Conditional $\cdf$ estimator with regression residuals\label{lem:condCDFEstResids}]

Let $(U_1,V_1),\dots,(U_n,V_n)$ be independent and identically distributed random variables with $U_1,\dots,U_n\in\mathbb{R}$. Partition the training indices $\{1,\dots,n\}$ into two non-intersecting splits, $\I_1$ and $\I_2$, respectively. Let $\Tcal_j:=\{(U_i,V_i)\}_{i\in\I_j}$ ($j=1,2$). Denote $\mu(v):=\E{U_i\mid V_i=v}$, and let $\hat{\mu}(v)$ denote its estimator based on $\Tcal_1$. Denote ${R}_i:=U_i-\mu(V_i)$ along with its approximation given by $\hat{R}_i:=U_i-\hat{\mu}(V_i)$ across $i=1,\dots,n$. For $u\in\text{support}(U_i)$ and $v\in\text{support}(V_i)$, denote $F(u\mid v):=\prc{U\leq u}{V=v}$. Consider its estimator given by:
\[
\hat{F}_n(u\mid v):=\frac{1}{|\I_2|}\sum_{i\in\I_2}\indic{ \hat{\mu}(v)+ \hat{R}_i\leq u },
\]
where $|\I_2|$ is the number of indices in $\I_2$. 

If $R_i\indep V_i$, then conditional on $\Vmat=(V_i;i\in\I_1)$, we have that for any $t\geq 0$, where $t$ is possibly dependent on $\Vmat$, that:
\[\begin{aligned}
&\sup_{u,v}\abs{ \hat{F}(u\mid v)-F(u\mid v)  }\\
\leq\ &\sup_r  \curl{\prc{ r<R_i\leq r+2t }{\Vmat} \vee \prc{ r-2t<R_i\leq r }{\Vmat} }+\curl{\frac{\log(4/\alpha)}{2} }^{1/2}|\I_2|^{-\frac{1}{2}} \\
\end{aligned}\]
with probability at least
\[
1-\alpha/2-{\rm pr}\curl{ \sup_v\abs{ \hat{\mu}(v)-\mu(v) }\geq t \mid \Vmat}. 
\]

\begin{proof}

Below, let the probability statements be with respect to $i\not\in\I_1$, where $\I_1$ is the set of indices used to train $\hat{\mu}$. This is important to note as we will use that $R_i\indep (\Tcal_1,\Vmat)$ later. Consider that:
\begin{equation}\begin{aligned}\label{eqn:controllingResidDistrReg}
&\sup_{u,v}\abs{ \hat{F}_{n}(u\mid v)-F(u\mid v) }\\
=\ &\sup_{u,v}\abs{ \hat{F}_{n}(u\mid v)-{\rm pr}\curl{ \mu(v)+R_i\leq u } }\\
\leq\ &\sup_{u,v}\abs{ \hat{F}_{n}(u\mid v)-{\rm pr}\curl{ \hat{\mu}(v)+\hat{R}_i\leq u \mid \Tcal_1} }\\
&+\sup_{u,v}\abs{ {\rm pr}\curl{ \hat{\mu}(v)+\hat{R}_i\leq u \mid \Tcal_1}-{\rm pr}\curl{ \mu(v)+R_i\leq u } }\\
\overset{(i)}{=}\ &\sup_{u,v}\abs{ \frac{1}{|\I_2|}\sum_{ i\in\I_2 }\indic{ U_i-\hat{\mu}(V_i)\leq u-\hat{\mu}(v) }-{\rm pr}\curl{ U_i-\hat{\mu}(V_i)\leq u-\hat{\mu}(v) \mid \Tcal_1} }\\
&+\sup_{u,v}\abs{ {\rm pr}\sqbrack{ {R}_i\leq u-\mu(v)+\curl{\hat{b}(V_i)-\hat{b}(v)} \mid \Tcal_1}-{\rm pr}\curl{ R_i\leq u-\mu(v) } }.\\
\end{aligned}\end{equation}
For any $v$, define:
\begin{equation}\label{eqn:biasRegrFunc}
\hat{b}(v):=\hat{\mu}(v)-\mu(v),
\end{equation}
a term that characterizes the bias in $\hat{\mu}(v)$ for any $v\in\text{support}(V_i)$. In equality (i), we used the identities $\hat{\mu}(v)=\mu(v)+\hat{b}(v)$ and $\hat{R}_i=R_i-\hat{b}(V_i)$, which are a consequence of the definition of $\hat{b}(\cdot)$.

Now, denote:
\[
A:=\sup_{u,v}\abs{ \frac{1}{|\I_w|}\sum_{ i\in \I_2 }\indic{ U_i-\hat{\mu}(V_i)\leq u-\hat{\mu}(v) }-{\rm pr}\curl{ U_i-\hat{\mu}(V_i)\leq u-\hat{\mu}(v) \mid \Tcal_1} }
\]
and
\[
B:=\sup_{u,v}\abs{ {\rm pr}\sqbrack{ {R}_i\leq u-\mu(v)+\curl{\hat{b}(V_i)-\hat{b}(v)} \mid \Tcal_1}-{\rm pr}\curl{ R_i\leq u-\mu(v) } }.
\]
Due to \eqref{eqn:controllingResidDistrReg} and Lemma \ref{lem:bndProbOfSum}:
\begin{equation}\begin{aligned}\label{eqn:devBoundResidApproach}
&{\rm pr}\curl{\sup_{u,v}\abs{ \hat{F}_{n}(u\mid v)-F(u\mid v) }> a+b \mid \Vmat}\\
\leq\ &\prc{ A+B> a+b }{\Vmat}\\
\leq\ &\prc{A> a}{\Vmat}+\prc{B> b}{\Vmat}.
\end{aligned}\end{equation}
We control the two terms in the latter inequality separately in \eqref{eqn:termALemma} and \eqref{eqn:termBLemma} below. We also explain the choices 
\[
a\ \dot{=}\ \curl{ 2^{-1}(|\I_2|)^{-1}\log(4/\alpha) }^{1/2}
\]
and 
\[
b\ \dot{=}\ \sup_r  \curl{\prc{ r<R_i\leq r+2t }{\Vmat} \vee \prc{ r-2t<R_i\leq r }{\Vmat} }.
\]
From these choices, we get the desired conclusion:
\[
\sup_{u,v}\abs{ \hat{F}_{n}(u\mid v)-F(u\mid v) }\leq a+b.
\]
with probability at least $1-\alpha/2-{\rm pr}\curl{\sup_v\abs{\hat{\mu}(v)-\mu(v) }\mid \Vmat}$.

For the term $\prc{A>a}{\Vmat}$ in \eqref{eqn:devBoundResidApproach}, notice that:
\[\begin{aligned}\label{eqn:conrollingResidDistrRegTerm1}
A=\ &\sup_{r}\abs{ \frac{1}{|\I_2|}\sum_{ i\in \I_2 }\indic{ U_i-\hat{\mu}(V_i)\leq r }-{\rm pr}\curl{ U_i-\hat{\mu}(V_i)\leq r \mid \Tcal_1} },\\
\end{aligned}\]
where the supremum with respect to $r=u-\hat{\mu}(v)$ is taken in the support of $u-\hat{\mu}(v)$. 

The quantity $\hat{\mu}(v)$ is random, even if $v$ is not, because it is an estimator based on $\Tcal_1$. However, conditional on $(\Tcal_1,\Vmat)$, $r=u-\hat{\mu}(v)$ is a constant. Importantly, conditional on $(\Tcal_1,\Vmat)$, we have that $\hat{R}_i=U_i-\hat{\mu}(V_i)$ is an independent and identically distributed random variable across $i\not\in \I_1$. This means that the estimator
\[
\frac{1}{ |\I_2| }\sum_{ i\in \I_2 }\indic{ U_i-\hat{\mu}(V_i)\leq r }
\]
for 
\[
{\rm pr}\curl{ U_i-\hat{\mu}(V_i)\leq r \mid \Tcal_1,\Vmat}
\]
satisfies the independent and identically distributed sample condition for the Dvoretzky–Kiefer–Wolfowitz inequality \citep{DKWIneq1956,massartDKWIneq1990,NaamanDKWIneq2021}. The Dvoretzky–Kiefer–Wolfowitz inequality gives:

\[\begin{aligned}
&\prc{ A > a }{\Tcal_1,\Vmat}\\
=\ &{\rm pr}\sqbrack{ \sup_{r}\abs{ \frac{1}{|\I_2|}\sum_{ i\in\I_2 }\indic{ U_i-\hat{\mu}(V_i)\leq r }-{\rm pr}\curl{ U_i-\hat{\mu}(V_i)\leq r \mid \Tcal_1} }>a \mid \Tcal_1,\Vmat}\\
{\leq}\ & 2e^{-2|\I_2| a^2}.
\end{aligned}\]
Noting that the upper bound does not depend on $(\Tcal_1,\Vmat)$, we get due to law of total expectation that:
\begin{equation}\begin{aligned}\label{eqn:termALemma}
&\prc{ A>a }{\Vmat}&{\leq}\ & 2e^{-2|\I_2| a^2} \\
\end{aligned}\end{equation}
and
\[\begin{aligned}
&\pr{ A> a }&{\leq}\ & 2e^{-2|\I_2|a^2}. \\
\end{aligned}\]
We will take
\[
a\ \dot{=}\ \curl{ 2^{-1}(|\I_2|)^{-1}\log(4/\alpha) }^{1/2},
\]
so that we are guaranteed $\pr{A > a}\leq\alpha/2$ and $\prc{A> a}{\Vmat}\leq \alpha/2$.

That covers the first important term $\prc{A>a}{\Vmat}$ in Equation \ref{eqn:devBoundResidApproach}. Now consider the second term, $\prc{B>b}{\Vmat}$. Here, 
\[
B:=\sup_{u,v}\abs{ {\rm pr}\sqbrack{ {R}_i\leq u-\mu(v)+\curl{\hat{b}(V_i)-\hat{b}(v)} \mid \Tcal_1}-{\rm pr}\curl{ R_i\leq u-\mu(v) } }
\]
is a random variable with respect to $\Tcal_1$. For $t$, possibly a function of $\Vmat$, consider the following event with respect to the random data in $\Tcal_1$:
\[
E(t):=\left\{\sup_v\abs{ \hat{\mu}(v)-\mu(v) }\leq t\right\}.
\]
The event $E(t)$ concerns the deviation between $\hat{\mu}$ and $\mu$ uniformly across the support of $V_i$. Conditional on the event $E(t)$ and $\Vmat$, we have that the regression bias term defined in \eqref{eqn:biasRegrFunc} satisfies $-t\leq\hat{b}(V_i)\leq t$, in spite of $V_i$ being random, along with $-t\leq\hat{b}(v)\leq t$ for any $v$. This means that conditionally on $\curl{E(t),\Vmat}$:
\begin{equation}\label{eqn:boundedBias}
-2t\leq\hat{b}(V_i)-\hat{b}(v)\leq 2t
\end{equation}
almost surely. 

Now let $E(t)^C$ denote the complement of $E(t)$. Consider that:
\begin{equation}\begin{aligned}\label{eqn:boundOnSecondQuantityResidThm}
\prc{B > b}{\Vmat}&=\ &{\rm pr}\curl{B > b \mid E(t),\Vmat}{\rm pr}\curl{E(t)\mid \Vmat}+{\rm pr}\curl{B > b\mid E(t)^C,\Vmat}{\rm pr}\curl{E(t)^C\mid \Vmat}\\
&\leq\ &{\rm pr}\curl{B > b \mid E(t),\Vmat}+{\rm pr}\curl{ \sup_v\abs{ \hat{\mu}(v)-\mu(v) }>t \mid \Vmat},\\
\end{aligned}\end{equation}
by the law of total probability, the fact that $0\leq \prc{\cdot}{\Vmat},{\rm pr}{\cdot\mid E(t),\Vmat}\leq 1$, along with the definition of $E(t)^C$.

For a random event $D\indep R_i$ (recall, $i\not\in\I_1$), consider the four-argument function:
\[
{\nu}(r,s,D,\Vmat):=\abs{ \prc{ {R}_i\leq r+s }{D,\Vmat}-\pr{ R_i\leq r } }.
\]
For $D=\emptyset$, we write:
\[
{\nu}(r,s,\emptyset,\Vmat):=\abs{ \prc{ {R}_i\leq r+s }{\Vmat}-\pr{ R_i\leq r } }.
\]
Fixing $(r,D,\Vmat)$, we can see that $\nu(r,s,D,\Vmat)$ is an increasing function with respect to $s$ in the domain $[0,2t]$, and it is a decreasing function with respect to $s$ in the domain $[-2t,0)$. From this, it follows that for all $-2t\leq s\leq 2t$:
\[
{\nu}(r,s,D,\Vmat)\ \leq\ {\nu}(r,-2t,D,\Vmat)\vee{\nu}(r,2t,D,\Vmat).
\]
Based on \eqref{eqn:boundedBias} and this property of $\nu(r,s,D,\Vmat)$, it follows that conditional on $E(t)$ and $\Vmat$ we have:
\[
0\leq B\leq \sup_{u,v}\curl{\nu(u-\mu(v),-2t,E(t),\Vmat)\vee\nu(u-\mu(v),2t,E(t),\Vmat)}
\]
almost surely. We have further that
\[\begin{aligned}
&\sup_{r}\left\{\nu(r,-2t,\emptyset,\Vmat)\vee\nu(r,2t,\emptyset,\Vmat)\right\}\\
=\ &\sup_{u,v}\curl{\nu(u-\mu(v),-2t,E(t),\Vmat)\vee\nu(u-\mu(v),2t,E(t),\Vmat)}.
\end{aligned}\]
This is based on the properties of the supremum, the fact that $t$ is fixed conditional on $\Vmat$, and because $R_i\indep E(t)$. So we can re-write that conditional on $E(t)$ and $\Vmat$:
\[
0\leq B\leq \sup_{r}\left\{\nu(r,-2t,\emptyset,\Vmat)\vee\nu(r,2t,\emptyset,\Vmat)\right\}
\]
almost surely. We will strategically set:
\[
b\ \dot{=}\ \sup_{r}\left\{\nu(r,-2t,\emptyset,\Vmat)\vee\nu(r,2t,\emptyset,\Vmat)\right\}.
\]
Moreover, $R_i\indep \Vmat$ means that $\pr{R_i\leq r}=\prc{R_i\leq r}{\Vmat}$, so we have:
\[
b\ =\ \sup_r  \curl{\prc{ r<R_i\leq r+2t }{\Vmat} \vee \prc{ r-2t<R_i\leq r }{\Vmat} }.
\]
With this choice of $b$, we have that ${\rm pr}\curl{B>b\mid E(t),\Vmat}=0$. Using this in \eqref{eqn:boundOnSecondQuantityResidThm}, we get that:

\begin{equation}\label{eqn:termBLemma}
\prc{B>b}{\Vmat}\leq {\rm pr}\curl{ \sup_v\abs{ \hat{\mu}(v)-\mu(v) }>t \mid \Vmat},
\end{equation}
as we wanted.

\end{proof}

\end{lemma}

\subsection{The proof of Corollary \ref{cor:effCondBounds}\label{append:proofOfCorEffCondBounds}}

\begin{proof}

Denote $f_w(r)$ as the marginal density of $R_i(w)$ ($w=0,1$). The key here is that
\[\begin{aligned}
&\sup_r\sqbrack{{\rm pr}\curl{r< R_i(w)\leq r+2t_w\mid \Xmat}\vee{\rm pr}\curl{r-2t_w<R_i(w)\leq r\mid \Xmat}}\\
\leq\ & 2t_w\max_{w=0,1}\sup_r f_w(r),\\
\end{aligned}\]
using that $\int_a^b h(u) du\leq \sup_u |h(u)| |b-a|$ for any integrable function $h$. We also used that $R_i(w)\indep \Xmat$. Under the regularity condition that the density $f_w(r)$ is non-negative and bounded away from infinity, we set $t_w\ \dot{ = }\ g_{n,\mathcal{F}}$ ($w=0,1$). This inequality and Proposition \ref{prop:mainResRegrResids} allow us to arrive at the desired result.
\end{proof}

\subsubsection{The proof of Proposition \ref{prop:CIsGaussResids}\label{append:proofCIsGaussResids}}

\begin{proof}[Proof of Proposition \ref{prop:CIsGaussResids}]

Let $\Psib_w\in R^{ |\I_1\cap S_w|\times d}$ be such that its rows are formed by stacking $\Psi_w(X_i)^T$ for each $i\in\I_1\cap S_w$. Further, let $\Yvec_w\in R^{|\I_1\cap S_w|\times 1}$ contain the corresponding observed outcome $Y_i$ for each $i\in\I_1\cap S_w$. The ordinary least squares estimator for the coefficient vector $\beta_w$ is
\[
\hat{\beta}_w=(\Psib_w^T\Psib_w)^{-1}\Psib_w^T\Yvec_w.
\]
We have that 
\[\begin{aligned}
&\sup_x\abs{\hat{\mu}_w(x)-\mu_w(x)}&=\ &\sup_x\abs{\Psi_w^T(x)(\hat{\beta}_w-\beta_w)}\\
&&\leq\ &\norm{\hat{\beta}_w-\beta_w}_2\\
&&\leq\ &\matnorm{ \sigma_w (\Psib_w^T\Psib_w)^{-1/2} }_{op}\norm{\frac{1}{\sigma_w}(\Psib_w^T\Psib_w)^{1/2}(\hat{\beta}_w-\beta_w) }_2.
\end{aligned}\]
The first inequality is based on H\"older's inequality, while the second inequality is again due to H\"older's inequality and the definition of the operator norm. Conditional on $\Psib_w$, the d-dimensional sampling distribution for $\hat{\beta}$ is: 
\[
\hat{\beta}_w|\Psib_w\sim\mathcal{N}(\beta_w,\sigma^2_w(\Psib_w^T\Psib_w)^{-1}).
\]
This is based on a standard argument in low dimensional linear regression given the residual distribution assumption. This further implies that conditionally on $\Psib_w$,
\[
\matnorm{ \sigma_w(\Psib_w^T\Psib_w)^{-1/2} }_{op}^2\norm{\frac{1}{\sigma_w}(\Psib_w^T\Psib_w)^{1/2}(\hat{\beta}_w-\beta_w) }_2^2\mid\Psib_w \sim \sigma_w^2 \matnorm{ (\Psib_w^T\Psib_w)^{-1/2} }_{op}^2\times\chi_{ d }^2,
\]
a distribution generated by taking a chi-squared distributed random variable ($d$ degrees of freedom) and multiplying it by a factor of $\sigma^2_w\matnorm{ (\Psib_w^T\Psib_w)^{-1/2} }_{op}^2$. This is because 
\[
\frac{1}{\sigma_w}(\Psib_w^T\Psib_w)^{1/2}(\hat{\beta}_w-\beta_w)\mid\Psib_w\sim\mathcal{N}_d(0,\mathbb{I})\implies\norm{\frac{1}{\sigma_w}(\Psib_w^T\Psib_w)^{1/2}(\hat{\beta}_w-\beta_w)}_2^2\mid\Psib_w\sim\chi^2_{ d }.
\]
Let $V$ have a $\chi^2_{ d }$ distribution conditional on $\Psib_w$, and denote $v_{d,\alpha}$ as the $(1-\alpha/2)$th quantile of $V$'s distribution. It follows that setting
\[
t_{w,\alpha}\ \dot{=}\ \parenth{v_{d,\alpha}}^{1/2}\sigma_w\matnorm{ (\Psib_w^T\Psib_w)^{-1/2} }_{op}
\]
implies that
\begin{equation}\begin{aligned}\label{eqn:confLevelPart}
&2\max_{w=0,1}{\rm pr}\curl{ \sup_x\abs{ \hat{\mu}_w(x)-\mu_w(x) }>t_{w,\alpha} \mid \Xmat}&\leq\ &2\max_{w=0,1}{\rm pr}\curl{V>\frac{t^2_{w,\alpha}}{ \sigma^2_w\matnorm{ (\Psib_w^T\Psib_w)^{-1/2} }_{op}^2 } \mid \Psib_w }\\
&&=\ &\alpha.
\end{aligned}\end{equation}
Now, consider that the distributional assumption on $R_i(w)$ implies the following. We have:
\begin{equation}\begin{aligned}\label{eqn:margErrorPart}
&2\max_{w=0,1}\sup_r{\rm pr}\curl{ r<R_i(w)\leq r+2t_{w,\alpha} }\vee{\rm pr}\curl{ r-2t_{w,\alpha}<R_i(w)\leq r }&\\
\overset{(i)}{=}\ &2\max_{w=0,1}\sup_r{\rm pr}\curl{ r<R_i(w)\leq r+2t_{w,\alpha} \mid \Psib_w}\vee{\rm pr}\curl{ r-2t_{w,\alpha}<R_i(w)\leq r \mid \Psib_w}&\\
\overset{(ii)}{=}\ &2\max_{w=0,1}\sup_r{\rm pr}\curl{ r<R_i(w)\leq r+2t_{w,\alpha} \mid \Psib_w}&\\
\overset{(iii)}{=}\ &2\max_{w=0,1}{\rm pr}\curl{ -t_{w,\alpha}<R_i(w)\leq t_{w,\alpha} \mid \Psib_w}&\\
\overset{(iv)}{=}\ &2\max_{w=0,1}{\rm pr}\curl{ -\parenth{v_{d,\alpha}}^{1/2} \matnorm{ (\Psib_w^T\Psib_w)^{-1/2} }_{op} <R_i(w)/\sigma_w\leq \parenth{v_{d,\alpha}}^{1/2} \matnorm{ (\Psib_w^T\Psib_w)^{-1/2} }_{op} \mid \Psib_w}&\\
\overset{(v)}{=}\ &2\max_{w=0,1}\sqbrack{\Phi\curl{ \parenth{v_{d,\alpha}}^{1/2} \matnorm{ (\Psib_w^T\Psib_w)^{-1/2} }_{op}}-\Phi\curl{ -\parenth{v_{d,\alpha}}^{1/2} \matnorm{ (\Psib_w^T\Psib_w)^{-1/2} }_{op} } }.&\\
\end{aligned}\end{equation}
Here, $\Phi$ denotes the $\cdf$ for the standard normal distribution. Equality (i) holds due to the assumption that $R_i(w)\indep \Lambda_w$. Equality (ii) holds due to the symmetry of the Gaussian density around its mean, which is zero in the case of $R_i(w)$ conditional on $\Lambda_w$. Moreover, (iii) holds because the biggest slice of area under the normal density of width $2t_{w,\alpha}$ is the one centered at its mean. Next, (iv) holds due to our choice of $t_{w,\alpha}$, while (v) holds since $R_i(w)/\sigma_w$ follows a standard normal distribution. \\

\noindent{}The final conclusion in Proposition \ref{prop:CIsGaussResids} follows by applying Proposition \ref{prop:mainResRegrResids} with \eqref{eqn:confLevelPart} and \eqref{eqn:margErrorPart}.

\end{proof}

\subsection{The proof of Proposition \ref{prop:boolFrechVsMakarov}\label{append:proofPropBoolVsMak}}

\begin{proof}[Proof of Proposition \ref{prop:boolFrechVsMakarov}]

The first claim is immediate from Lemma \ref{lem:makBounds} and our definition of $\theta^L(\delta)$ and $\theta^U(\delta)$ in Section \ref{sec:margITEMakBnds}. We also use the fact that $\pns$ is the same as $\pibt$ when $0\leq\delta<1$.

We now show why the second part of our claim holds. We will use that $\cdf$s for binary potential outcomes satisfy:
\[
F_w(y)=\begin{cases}0&\text{ if }y<0\\{\rm pr}\curl{Y_i(w)=0}&\text{ if }0\leq y<1\\ 1&\text{ if } y\geq1\end{cases}; w=0,1.
\]

For $0\leq\delta<1$, the Makarov lower bound on $\pibt$ is:
\[\begin{aligned}
&\theta^L(\delta)&=\ &-\min\sqbrack{ \inf_{y=0,1}\curl{F_1(y+\delta/2)-F_0(y-\delta/2) },0 }\\
&&=\ &-\min\curl{ F_1(\delta/2)-F_0(-\delta/2),F_1(1+\delta/2)-F_0(1-\delta/2) ,0 }\\
&&=\ &\begin{cases}-\min\sqbrack{ {\rm pr}\curl{Y_i(1)=0}-{\rm pr}\curl{Y_i(0)=0},0  }&\text{ if }\delta=0\\
-\min\sqbrack{ {\rm pr}\curl{Y_i(1)=0},1-{\rm pr}\curl{Y_i(0)=0},0  }&\text{ if }0<\delta<1\end{cases}\\
&&=\ &\begin{cases}-\min\sqbrack{ 1-{\rm pr}\curl{Y_i(1)=1}-{\rm pr}\curl{Y_i(0)=0},0  }&\text{ if }\delta=0\\0&\text{ if }0<\delta<1\end{cases}\\
&&=\ &\begin{cases}\max\sqbrack{ {\rm pr}\curl{Y_i(1)=1}+{\rm pr}\curl{Y_i(0)=0}-1,0  }&\text{ if }\delta=0\\ 0&\text{ if }0<\delta<1\end{cases}.\\
\end{aligned}\]

It follows that
\[
\sup_{0\leq\delta<1}\theta^L(\delta)=\max\sqbrack{ {\rm pr}\curl{Y_i(1)=1}+{\rm pr}\curl{Y_i(0)=0}-1,0  },
\]
as we wanted.

Similarly, for $0\leq\delta<1$, the Makarov upper bound on $\pibt$ is:
\[\begin{aligned}
&\theta^U(\delta)&=\ &1-\max\sqbrack{ \sup_{y=0,1}\curl{F_1(y+\delta/2)-F_0(y-\delta/2) },0 }\\
&&=\ &1-\max\sqbrack{F_1(\delta/2)-F_0(-\delta/2) ,F_1(1+\delta/2)-F_0(1-\delta/2) ,0 }\\
&&=\ &\begin{cases}1-\max\sqbrack{{\rm pr}\curl{Y_i(1)=0}-{\rm pr}\curl{Y_i(0)=0} ,0  }&\text{ if }\delta=0\\ 1-\max\sqbrack{{\rm pr}\curl{Y_i(1)=0} ,1-{\rm pr}\curl{Y_i(0)=0}  }&\text{ if }0<\delta<1\end{cases}\\
&&=\ &\begin{cases}1-\max\sqbrack{1-{\rm pr}\curl{Y_i(1)=1}-{\rm pr}\curl{Y_i(0)=0} ,0  }&\text{ if }\delta=0\\1-\max\sqbrack{{\rm pr}\curl{Y_i(1)=0} ,{\rm pr}\curl{Y_i(0)=1}  }&\text{ if }0<\delta<1\end{cases}\\
&&=\ &\begin{cases}\min\sqbrack{{\rm pr}\curl{Y_i(1)=1}+{\rm pr}\curl{Y_i(0)=0} ,1  }&\text{ if }\delta=0\\\min\sqbrack{{\rm pr}\curl{Y_i(1)=1} ,{\rm pr}\curl{Y_i(0)=0}  }&\text{ if }0<\delta<1\end{cases}\\
\end{aligned}\]

It follows that
\[
\inf_{0\leq\delta<1 }\theta^U(\delta)=\min\sqbrack{{\rm pr}\curl{Y_i(1)=1} ,{\rm pr}\curl{Y_i(0)=0}  },
\]
as we wanted.

With essentially the same reasoning, we have the analogous claims for \[{\rm pr}\curl{Y_i(1)=1,Y_i(0)=0 \mid X_i=x}.\]

\end{proof}

\end{document}